\documentclass[11pt]{article}
\usepackage{customstyle}
\usepackage{hyperref}
\usepackage{subfiles} 
\usepackage[margin=1in]{geometry}

\newcommand{\shortpage}{\enlargethispage{-\baselineskip}}
\usetikzlibrary{matrix,arrows.meta,decorations.pathreplacing,calc,positioning}
\title{A Framework for Differential Privacy Against Timing Attacks}
\author{Zachary Ratliff\thanks{Harvard University \& OpenDP. Email: \texttt{zacharyratliff@g.harvard.edu}. Supported in part by Cooperative Agreement CB20ADR0160001 with the Census Bureau, and in
part by Salil Vadhan’s Simons Investigator Award}
\and
Salil Vadhan\thanks{Harvard University \& OpenDP. Email: \texttt{salil\_vadhan@harvard.edu}}
}
\date{September 2024}

\begin{document}

\maketitle
\begin{abstract}
The standard definition of differential privacy (DP) ensures that a mechanism's \emph{output} distribution on adjacent datasets is indistinguishable. However, real-world implementations of DP can, and often do, reveal information through their \emph{runtime} distributions, making them susceptible to timing attacks. In this work, we establish a general framework for ensuring differential privacy in the presence of timing side channels. We define a new notion of \emph{timing privacy}, which captures programs that remain differentially private to an adversary that observes the program's runtime in addition to the output. Our framework enables chaining together component programs that are \emph{timing-stable} followed by a random delay to obtain DP programs that achieve timing privacy.  Importantly, our definitions allow for measuring timing privacy and output privacy using different privacy measures. We illustrate how to instantiate our framework by giving programs for standard DP computations in the RAM and Word RAM models of computation.  Furthermore, we show how our framework can be realized in code through a natural extension of the OpenDP Programming Framework. 
\end{abstract}
\newpage
\tableofcontents
\newpage
\section{Introduction}
The framework of differential privacy (DP) \cite{dwork2006calibrating} is used extensively for computing privacy-preserving statistics over sensitive data. A differentially private algorithm has the property that \emph{close} inputs map to \emph{indistinguishable} output distributions. Here, ``close'' and ``indistinguishable'' are often given by various metrics and probability distance measures respectively (\S\ref{sec:prelims}). For example, algorithms that compute dataset statistics commonly require that adding or removing any individual's data does not change the probability of outputting any given value by more than a constant factor. More formally, the definition of differential privacy is given as follows.

\begin{definition}[Differential Privacy \cite{dwork2006calibrating, dwork2006our}]
Let $M: \mathcal{X} \rightarrow \mathcal{Y}$ be a randomized function.
We say $M$ is $(\varepsilon,\delta)$-\emph{differentially private} if for every pair of adjacent datasets $x$ and $x'$ and every subset $S \subseteq \mathcal{Y}$
\begin{align*}
\Pr[M(x) \in S] \le e^\varepsilon\cdot \Pr[M(x') \in S] + \delta
\end{align*}
\end{definition}

\noindent When $\delta = 0$ we say that $M$ satisfies \emph{pure} differential privacy. Intuitively, this property guarantees that the result of a statistical analysis is essentially the same regardless of the presence or absence of any individual's data. Notable examples of differential privacy in practice include the US Census Bureau~\cite{abowd2018us, machanavajjhala2008privacy},  Apple~\cite{greenberg2016apple}, Facebook~\cite{messing2020facebook}, and Google~\cite{aktay2020google} who have deployed DP to expand access to sensitive data while protecting individual privacy. DP has also been applied to problems that are not explicitly about privacy-preserving statistics such as private machine learning~\cite{abadi2016deep, papernotscalable}, adaptive data analysis~\cite{bassily2016algorithmic, dwork2015preserving, neel2018mitigating}, anonymous messaging~\cite{lazar2018karaoke,tyagi2017stadium,van2015vuvuzela}, and distributed analytics~\cite{roth2021mycelium, roth2020orchard}.

Unfortunately, implementing differential privacy faithfully is a tricky business. Algorithms that are proven on paper to achieve differential privacy are eventually implemented and ran on hardware with various resource constraints. These constraints can lead to discrepancies between the computational model used in the privacy proof and the actual implementation, potentially invalidating the algorithm's privacy guarantees. For example, many mathematical proofs in the DP literature assume exact arithmetic over the real numbers while their implementations instead use finite-precision floating-point or integer arithmetic. Such inconsistencies have led to faulty implementations of textbook mechanisms that do not actually achieve differential privacy~\cite{casacuberta2022widespread, jin2022we, mironov2012significance}, and prompted new research on differential privacy in finite models of computation~\cite{balcer2018differential, canonne2020discrete}. In a similar vein, real-world implementations of differential privacy are often exposed to observable \emph{side channels} that are not modeled in the definition of differential privacy. For instance, in the online query setting, where an analyst submits a query and receives a differentially private response, the analyst observes not only the query's output but also the query's execution time. Leveraging the timing of operations to learn otherwise protected information is known as a \emph{timing attack}, and prior work has demonstrated that such attacks can be used to violate differential privacy. For example, Haeberlen, Pierce, and Narayan showed how user-defined queries can be used to leak the presence of records in a sensitive dataset via their running time~\cite{haeberlen2011differential}. More recently, Jin, McMurtry, Rubinstein, and Ohrimenko discovered that the runtime of discrete sampling algorithms (specifically, the Discrete Laplace~\cite{ghosh2012universally} and Discrete Gaussian~\cite{canonne2020discrete}) can reveal the magnitude of their sampled noise values~\cite{jin2022we}. Leaking the sampled noise value can result in a violation of DP since it can be used to remove the noise and hence the privacy protection from the output of the algorithm.

Despite some growing awareness of timing attacks on differential privacy, prior work has focused on ad hoc mitigation strategies that typically apply to a single class of differentially private mechanisms\footnote{The differential privacy literature uses the word ``mechanism" to refer to a randomized algorithm. Since our work focuses primarily on actual implementations of such mechanisms within a computational model, we will instead use the term ``programs" throughout the paper.} or database system~\cite{haeberlen2011differential, balcer2018differential, awan2023privacy}. The most commonly suggested strategy has been to enforce constant-time program execution. However, enforcing constant-time execution inherently limits the amount of data that the program can process. The program must either truncate the input (e.g., by performing subsampling) or reveal an upper bound on the accepted dataset size. Constant-time programs also come with a significant performance cost, as fast computations must be padded to take the same amount of time as the slowest possible computation. Additionally, constant-time program execution is overkill. It prevents {\em all} information leakage from the runtime, whereas differential privacy only requires that executions over neighboring inputs are indistinguishable. 

To avoid enforcing constant-time program execution, Haeberlen et al. suggested making the program's runtime differentially private by adding random delay before returning the output~\cite{haeberlen2011differential}. This approach mimics that of adding noise to the program's output to achieve differential privacy. However, this proposal is incomplete because it does not specify 
\begin{enumerate}[label=(\arabic*)]
    \item What it means for a program's runtime to be differentially private? and 
    \item How much delay, and from what distribution, is necessary to achieve that definition of privacy?
\end{enumerate}  

New definitions and theorems are needed to precisely reason about the privacy guarantees of such programs in the presence of timing side channels. 

\subsection{Our Results}
To reason about timing side channels, we need to reason about programs rather than functions. The behavior of programs includes both the program's output and runtime, which can depend on the program's execution environment. This might include, for example, the current state of cache memory or any resource contention caused by concurrent program execution (see Section~\ref{sub-sec:computational-models} below for more on execution environments). For this reason, we explicitly generalize the definition of differential privacy to this setting.
    
\noindent\begin{definition}[$(\varepsilon, \delta)$-Differentially Private Programs]
    \label{def:approx-dp-programs}
        Let $\mathcal{E}$ be the set of possible execution environments for a fixed computational model and let $P: \mathcal{X}\times\mathcal{E}\to \mathcal{Y}\times\mathcal{E}$ be a randomized program.
        We say $P$ is $(\varepsilon,\delta)$-\emph{differentially private} if for every pair of adjacent datasets $x$ and $x'$, every pair of input-compatible execution environments $\env, \env' \in \mathcal{E}$, and every subset $S \subseteq \mathcal{Y}$
        \[
        \Pr[\out(P(x, \env)) \in S] \le e^\varepsilon\cdot \Pr[\out(P(x', \env')) \in S] + \delta
        \]

        where $\out(P(x, \env))$ indicates $P's$ output value $y \in \mathcal{Y}$.
    \end{definition}
    
    \noindent Given this more generalized definition, we define a new notion of privacy with respect to timing attacks. Specifically, our definition asks that the running time of a program is differentially private \emph{conditioned on the output of the program}. For the special case of pure DP, our definition is as follows:

    \begin{definition}[$\varepsilon$-OC-Timing Privacy, special case of Def.~\ref{def:rtprivacy}]\label{def:rt-privacy-simple} Let $P: \mathcal{X} \times\mathcal{E} \rightarrow \mathcal{Y}\times\mathcal{E}$ be a (possibly) randomized program. Then we say that $P$ is $\varepsilon$-\emph{OC-timing-private} if for all adjacent $x, x' \in \mathcal{X}$, all pairs of input-compatible execution environments $\env, \env' \in \mathcal{E}$, all $y \in \supp(\out(P(x, \env))) \cap \supp(\out(P(x', \env')))$, and all $S \subseteq \mathcal{T}$
    \begin{flalign*}
                \Pr[\TP{x, \env} \in S |\out(P(x, \env)) = y] \le e^\varepsilon\cdot  \Pr[\TP{x', \env'} \in S | \out(P(x', \env')) = y]
    \end{flalign*}
        
        where $\TP{x, \env}$ denotes the running time of $P$ on input $x$ in execution environment $\env$, 
    $\mathcal{T}\subseteq\mathbb{R}_{\ge 0}$ represents the units in which we measure time (e.g., $\mathcal{T} = \mathbb{N}$ if we count the number of executed instructions), and $\out(P(x, \env))$ indicates the program's output.
    \end{definition}

    \noindent Intuitively, this definitional approach provides several benefits, one of which is that it enables measuring timing privacy using a different privacy measure than that used to measure output privacy. For example, one can be measured by approximate DP and the other with Rényi DP. This feature is useful since one generally compromises on utility to achieve output privacy, while one can instead compromise on efficiency (longer runtimes) to achieve timing privacy. In contrast, a previous definition of Ben Dov, David, Naor, and Tzalik~\cite{ben2023resistance} (which we rename\footnote{Ben Dov et al. described mechanisms that meet this criteria as \emph{differentially private time oblivious mechanisms}. We will instead use the more informative term \emph{jointly output/timing-private}.}) couples the leakage of the output and the runtime: 

    \begin{definition}[$(\varepsilon, \delta)$-Joint Output/Timing Privacy~\cite{ben2023resistance}, special case of Def.~\ref{def:joint-privacy}]
    \label{def:approx-joint-privacy}
         Let $P: \mathcal{X}\times\mathcal{E} \rightarrow \mathcal{Y}\times\mathcal{E}$ be a (possibly) randomized program. Then we say that $P$ is $(\varepsilon, \delta)$-\emph{jointly output/timing-private} if for all adjacent $x, x' \in \mathcal{X}$, all pairs of input-compatible execution environments $\env, \env' \in \mathcal{E}$, and for all $S \subseteq \mathcal{Y}\times\mathcal{T}$
    \begin{flalign*}
        \Pr[(\out(P(x, \env)), \TP{x, \env}) \in S] \le e^\varepsilon\cdot  \Pr[(\out(P(x', \env')), \TP{x', \env'}) \in S] + \delta
    \end{flalign*}
    \end{definition}
        
    This is the standard definition of DP applied to the joint random variable of the program's output and running time. We show that in the case of pure-DP ($\delta = 0$), joint output/timing privacy is equivalent to output privacy together with our definition of OC-timing privacy up to a constant factor in $\varepsilon$. More generally, our definition together with output privacy implies joint output/timing privacy when both output and timing privacy are measured using the same privacy measure (Lemma~\ref{lemma:timing-privacy-output}).  \\
    
    \noindent We establish a general framework for chaining together \emph{timing-stable programs} (programs where close inputs map to close runtime distributions) with \emph{timing-private delay programs} (programs that add random delay before releasing their output) such that the entire execution is \emph{timing-private}. DP software libraries such as OpenDP~\cite{gaboardi2020programming, Shoemate_OpenDP_Library} and Tumult Core~\cite{berghel2022tumult} support the concatenation of multiple data operations, such as clamping, summation, and adding noise from a distribution, into a unified mechanism. Our framework enables examining how adjacent datasets impact the execution time of modular components within such a chain. A single timing delay can then be applied before returning the mechanism's output to safeguard the entire chain against timing attacks. \\

    \noindent To illustrate how our framework can be instantiated, we provide concrete examples of RAM and Word RAM programs that satisfy our notion of timing privacy when we treat each RAM instruction as a single time step. These include mechanisms for computing a randomized response, private sums, and private means (Section~\ref{sec:chaining}).
    Notably, our mechanisms are often much more efficient than constant-time constructions. We also describe mechanisms that achieve timing privacy in the unbounded DP setting and have accuracy that matches their non-timing-private counterparts. \\

    \noindent To show the compatibility of our framework with existing libraries for differential privacy, we implement a proof-of-concept timing-privacy framework within the OpenDP library~\cite{Shoemate_OpenDP_Library}. 
    Our implementation supports the tracking of timing stability across dataset transformations, and includes an implementation of a timing-private delay program for delaying the output release to achieve timing-privacy.  This proof-of-concept is only meant to show how our framework can be easily implemented on top of existing frameworks, not to provide actual guarantees of privacy against an adversary who can measure physical time, which we leave as an important direction for future work. 

\subsection{Future Work}\label{sec:future-work}
We have shown how one can rigorously reason about timing privacy in idealized computational models (like the RAM and Word-RAM models) where every instruction is assumed to take the same amount of time. On physical CPUs, this assumption does not hold and the varying time for different operations has been used as the basis for past timing attacks~\cite{andrysco2015subnormal}. Nevertheless, we believe that our {\em framework} can be fruitfully applied to physical computations with suitably constrained execution environments. Importantly, our notion of timing stability does not require precise estimates of computing time, but only an {\em upper bound} on the time difference that one individual's data can make on a computation, which seems feasible to estimate in practice.

Another direction for future work is to design timing-private and timing-stable programs for a wider array of mechanisms and transformations from the differential privacy literature.  Some examples of interest would be the Discrete Gaussian Mechanism~\cite{canonne2020discrete} and forms of the Exponential Mechanism~\cite{mcsherry2007mechanism}, as well providing user-level privacy on datasets where a user may contribute many records, and node-level privacy on graphs.  Getting timing privacy for DP algorithms that involve superlinear-time operations such as sorting (e.g. approximate medians via the exponential mechanism~\cite{mcsherry2007mechanism}) or pairwise comparisons (e.g. the DP Theil-Sen linear regression algorithm~\cite{alabi2022differentially}) are also an intriguing challenge, because timing-stable programs must have $O(n)$ runtime (Lemma~\ref{lemma:timing-stable-implies-linear}).

Finally, from a theoretical point of view, it is interesting to explore the extent to which {\em pure} timing privacy is achievable.  Our examples of timing-private programs mostly come from adding random delays to timing-stable programs, which seems to inherently yield approximate DP (i.e. $\delta>0$).   Ben-Dov et al.~\cite{ben2023resistance} show some limitations of what can be achieved with pure DP, but their result does not rule out many useful DP computations (such as a RAM program that randomly samples
$O(1)$ records from a dataset and then carries out a constant-time DP computation on the subsample, which achieves perfect timing privacy even on unbounded-sized datasets).
\section{Preliminaries}
\label{sec:prelims}
\subsection{Models of Computation}\label{sub-sec:computational-models}
The runtime of a program is determined by its computational model. For instance, within the idealized RAM model, runtime is typically measured by the number of executed instructions until halting. Conversely, for a C program running on an x86 architecture, runtime might be quantified in CPU cycles, though the same program on identical hardware can exhibit varied runtimes due to differences in microcode patches \cite{stallings2003computer}. Moreover, runtime can be affected by non-deterministic factors like cache interference from concurrent executions on modern CPUs. Thus, we allow a program's output and runtime to depend on an execution environment $\env \in \mathcal{E}$, where an execution environment is defined to be any state stored by the computational model that affects program execution (e.g., the contents of uninitialized memory, the program's input, the initial values stored in CPU registers, etc.). We consider the program's input to be part of the program's execution environment and say that $\supp(P(x, \env)) = \emptyset$ if the input and environment are \emph{incompatible}.

\begin{definition}[Input/Environment Compatibility]
    Our computational models come with an input/environment compatibility relation, and when an environment $\env \in \mathcal{E}$ is incompatible with an input $x\in\mathcal{X}$ for program $P$, we define $\supp(P(x, \env)) = \emptyset$.
\end{definition}

The compatibility of a given $(x, \env)$ pair depends on the computational model. For example, the calling convention for C functions on modern x86 hardware require the EBP and ESP registers, which indicate the base and top of the stack, to specify the function's stack frame in memory. If these registers point to conflicting memory locations or do not contain the correct setup for the function's input arguments, the execution environment would be incompatible with the function's input. We give further examples of such incompatible input/environment pairs below when discussing the RAM model below.

\begin{definition}[Program Execution]
A randomized program $P : \mathcal{X}\times\mathcal{E}\to \mathcal{Y}\times\mathcal{E}$ takes in an input $x \in \mathcal{X}$, and an execution environment $\env \in \mathcal{E}$, and outputs a value $y \in \mathcal{Y}$ denoted by the random variable $\out(P(x, \env))$. The execution of $P$ may also induce changes on the environment, and we indicate the new state of the environment as the random variable $\out\env(P(x, \env))$. 
\end{definition}

\begin{definition}[Program Runtime]
 Suppose that our computational model's runtime is measured in units $\mathcal{T}\subseteq\mathbb{R}_{\ge 0}$ (e.g., $\mathcal{T} = \mathbb{N}$ if runtime measures the number of executed instructions). Then we denote the \emph{runtime} of program $P:\mathcal{X}\times\mathcal{E}\to\mathcal{Y}\times\mathcal{E}$ within execution environment $\env$ on input $x$ to be the random variable $T_P(x, \env) \in \mathcal{T}$.
\end{definition}

We can drop the environment $\env$ from our notation when a program's output and runtime is independent of the environment. We call such programs \emph{pure} to follow standard programming languages terminology, but not to be confused with pure DP. 

\begin{definition}[Output-Pure Programs]
\label{def:output-pure-programs}
     A program $P:\mathcal{X}\times\mathcal{E}\to\mathcal{Y}\times\mathcal{E}$ is \emph{output-pure} if there exists a (possibly randomized) function $f:\mathcal{X}\to\mathcal{Y}$ such that for all $x\in\mathcal{X}$ and all input-compatible $\env\in\mathcal{E}$, $f(x)$ is identically distributed to $\out(P(x, \env))$.
\end{definition}

\begin{definition}[Timing-Pure Programs]\label{def:timing-pure-programs}
 A program $P:\mathcal{X}\times\mathcal{E}\to\mathcal{Y}\times\mathcal{E}$ is \emph{timing-pure} if there exists a (possibly randomized) function $f:\mathcal{X}\to\mathcal{T}$ such that for all $x \in\mathcal{X}$, and all input-compatible $\env \in\mathcal{E}$, $f(x)$ is identically distributed to $\TP{x, \env}$.
\end{definition}

We remark that the standard definition of differential privacy implicitly assumes that programs are \emph{output-pure}, and in fact, its composition and post-processing theorems fail if that is not the case.  \\

\noindent\textbf{RAM Model}. In this work, we will use the idealized RAM and Word RAM models of computation to give proof-of-concept results on how our framework can be used. Applying our framework to physical hardware implementations is an important challenge for future work. The RAM model includes an infinite sequence of memory cells, each capable of holding an arbitrarily large natural number. The model stores variables in registers and supports a set of basic operations for performing arithmetic, logical, and memory operations. Programs also include instructions for conditional jumps (\textbf{if} CONDITIONAL \textbf{goto} LINE) which allow simulation of the standard \textbf{if}, \textbf{else}, and \textbf{while} expressions which we freely use.\footnote{We define condition matching in branching expressions to take $1$ instruction.} Additionally, many of our RAM programs are randomized and use a $\texttt{RAND}(n)$ instruction to sample a uniform random integer in $\{0, \dots, n\}$. We define the runtime of a RAM program to be the number of basic instructions executed before halting, which is indicated by a $\texttt{HALT}$ command. That is, the space of runtime values  $\mathcal{T}$ is equal to the natural numbers $\mathbb{N}$. The model includes a built-in variable $\texttt{input\_len}$ that indicates the length $|x|$ of the program's input $x$, as well as a built-in variable $\texttt{input\_ptr}$ that indicates where in memory the input resides. Additionally, every program writes their output to some location in memory indexed by the variable $\texttt{output\_ptr}$ and sets a variable $\texttt{output\_len}$ to indicate the length of the output.

RAM programs should not make assumptions about the values loaded in memory locations outside the locations where the input resides, and are responsible for initializing any such memory that they use. This specification aligns with the behavior observed in, for instance, C programs, wherein it cannot be presumed that newly allocated memory through a \texttt{malloc}\footnote{\texttt{malloc} is the standard C library function for dynamic memory allocation.} instruction is initialized to zero. We remark that this specification introduces an element of non-determinism to our model since a program may read and use uninitialized memory values. The space of environments $\mathcal{E}$ for a RAM program is therefore the set of possible memory configurations as well as the initial values stored by the built-in variables $\texttt{input\_ptr}, \texttt{input\_len}, \texttt{output\_ptr}, $ and $\texttt{output\_len}$. Some environments may not be relevant for a given input, e.g., the environment corresponding to all memory locations being uninitialized is not possible for programs whose input lengths are greater than zero. Similarly, compatible $(x, \env)$ pairs have the property that the memory locations of $\env$ indicated by the variables $\texttt{input\_ptr}$ and $\texttt{input\_len}$ should contain the input $x$. For all such incompatible pairs $(x, \env)$, $P(x, \env)$ is undefined and so we say that $\supp(P(x, \env)) = \emptyset$ in Definition~\ref{def:rt-privacy-simple} and elsewhere. Finally, most of our RAM programs are \emph{pure} and will not interact with uninitialized memory, and therefore both the program's output and running time will be independent of its execution environment for all compatible $(x, \env)$ pairs. \\

\begin{definition}[RAM Environment]
The \emph{environment} $\env$ of a RAM program is the infinite sequence $(v_{0}, v_{1}, \dots, )$ such that $M[i] = v_i$ for all $i$ along with the values stored by the built-in variables \texttt{input\_ptr}, \texttt{input\_len}, \texttt{output\_ptr}, and \texttt{output\_len}. 
\end{definition}

\begin{remark}
    A program's environment includes the program's input $x$ by definition, and similarly, a program's output environment $\out\env(P(x, \env))$ includes the program's output $\out(P(x, \env))$. \\
\end{remark}

\noindent\textbf{Word RAM.} We will use a more realistic $\omega$-bit Word RAM model for describing programs that operate on bounded-length values. The word size $\omega$ corresponds to the bitlength of values held in memory and variables, effectively capping the total addressable memory to $2^\omega$. This limitation arises because memory access, denoted by $M[\texttt{var}]$, relies on the size of $\texttt{var}$, which cannot exceed $2^\omega - 1$. We fix the word size $\omega$ up front which bounds the worst-case dataset size.  Therefore the model does not allow inputs to grow unboundedly. This constraint is principally driven by the fact that computers in the real world are equipped with finite memory and a pre-determined word size. Furthermore, allowing the word size to vary with the length of the input (e.g., it is standard in the algorithms literature $\omega = \Theta(\log n)$ and allow $n\to\infty$) introduces an additional side channel within the computational framework. In particular, the output itself is given as a sequence of $\omega$-bit words, which reveals information about the input length. For these reasons, we treat the word size as a separate public parameter which implies an upper bound on the input length. 


\subsection{Datasets, Distance Metrics, and Privacy Measures}
Differential privacy (and timing privacy, Definition~\ref{def:rt-privacy-simple}) is defined with respect to a dataset space $\mathcal{X}$. Typically, datasets consist of $n\ge 0$ records drawn from a row-domain $\mathcal{D}$, i.e., $\mathcal{X} = \bigcup_{n=0}^\infty \mathcal{D}^n$. We take elements of $\mathcal{D}$ as consisting of an individual's data. As such, $\mathcal{D}$ often consists of $d$-dimensional entries, for example, in tabular data where each individual has $d$ attributes. Thus, in the RAM model we take $\mathcal{D} = \mathbb{N}^d$ and in the Word RAM model $\mathcal{D} = \{0, \dots, 2^\omega - 1\}^d$. Since input lengths are bounded by $2^\omega$ in Word RAM, we only consider datasets of size at most some $n_{\text{max}} < 2^{\omega/d}$. 

Following the OpenDP Programming Framework~\cite{gaboardi2020programming}, which is also the basis of Tumult Core~\cite{berghel2022tumult}, our timing privacy and stability definitions work with arbitrary metrics on input data(sets) and arbitrary measures of privacy. The dataset metrics are arbitrary, user-defined metrics that are not required to satisfy the standard properties of non-negativity, symmetry, and triangle inequality as in the standard mathematical definition of a metric. Two common choices for dataset distance metrics are the \emph{Hamming distance} and \emph{insert-delete distance}.

\begin{definition}
    [Hamming Distance]\label{def:hamming-metric}
    Let $x, x' \in \mathcal{D}^*$ be datasets. The \emph{Hamming distance} denoted $\dham$ of $x$ and $x'$ is
    \begin{align*}
        \dham(x, x') = 
        \begin{cases}
            \#\{i : x_i \ne x'_i \} & \textrm{if } |x| = |x'| \\
            \infty & \textrm{otherwise}
        \end{cases}
    \end{align*}
\end{definition}

\begin{definition}[Insert-Delete Distance]
For $x \in \mathcal{D}^*$, an insertion to $x$ is an addition of an
element $z$ to a location in $x$ resulting in a new dataset $x' = [x_1, \dots, x_i, z, x_{i+1}, \dots, x_n]$.
Likewise, a deletion from $x$ is the removal of an element from some location within $x$, resulting in a new dataset $x' = [x_1, \dots, x_{i - 1}, x_{i + 1}, \dots, x_n]$. The \emph{insert-delete distance} denoted $\did$ of datasets $x, x' \in \mathcal{D}^*$ is the minimum number of insertion and deletion operations needed to change $x$ into $x'$.
\end{definition}

We focus our attention on ordered distance metrics since the running time of a program often depends on the ordering of its input data, e.g., if the algorithm starts by sorting its input. 

We say that two datasets $x$ and $x'$ are \emph{adjacent} with respect to dataset distance metric $d_\mathcal{X}$ if $d_{\mathcal{X}}(x, x') \le 1$. When $d_{\mathcal{X}}=\dham$, this notion of adjacency captures ``bounded differential privacy,'' where the dataset size $n$ is known and public.  When $d_\mathcal{X}=\did$, it captures ``unbounded differential privacy,'' where the dataset size itself may be unknown. We will also use what we call ``upper-bounded differential privacy,'' where we assume a known and public upper bound $n_{\textrm{max}}$ on the dataset size, i.e., $\mathcal{X} = \mathcal{D}^{\le n_{\textrm{max}}} = \bigcup_{n=0}^{n_{\textrm{max}}} \mathcal{D}^n$ with metric $d_\mathcal{X} = \did$ so that the exact dataset size is unknown and needs to be protected with DP. This, for example, arises in the $\omega$-bit Word RAM model of computation where $n_{\textrm{max}} \le 2^{\omega / d}$ for $d$-dimensional datasets. We note that these metrics are appropriate for tabular datasets where each record corresponds to one individual's data. Our framework can also be instantiated for other data domains and metrics (e.g. user-level DP in a dataset of events or graph DP with node privacy), but we leave designing algorithms for these settings as future work. 

A core feature of several DP systems and libraries is the ability to combine various component functions into more complex differentially private mechanisms~\cite{mcsherry2009privacy, haeberlen2011differential, gaboardi2020programming}. Analyzing the \emph{stability} of the various component functions is a useful method for understanding the privacy implications of arbitrarily combining the functions~\cite{mcsherry2009privacy}. The stability of a deterministic program characterizes the relationship between its input and output distances. Executing the program on ``close'' inputs will result in outputs that are also ``close.'' To define stability for randomized programs, we'll need the notion of couplings. 

\begin{definition}[Coupling]
A \emph{coupling} of two random variables $r$ and $r'$ taking values in sets $\mathcal{R}$ and $\mathcal{R}'$ respectively, is a joint random variable $(\Tilde{r}, \Tilde{r}')$ taking values in $\mathcal{R} \times \mathcal{R}'$ such that $\Tilde{r}$ has the same marginal distribution as $r$ and $\Tilde{r}'$ has the same marginal distribution as $r'$. 
\end{definition}

Couplings allow us to extend the notion of stability to randomized programs by measuring stability in terms of distributions rather than fixed outcomes. Specifically, the stability of a randomized program can be interpreted as a worst-case analogue of the Wasserstein\footnote{The Wasserstein distance between probability distributions $Q$ and $S$ on metric space $(\mathcal{Y}, d_\mathcal{Y})$ is defined as:
    \begin{align*}
        W_1(Q, S) &= \inf_{\gamma \in \Gamma} \mathbb{E}_{(q, s) \sim \gamma} d_\mathcal{Y}(q, s)
    \end{align*}
where $\Gamma$ is the set of all couplings of $Q$ and $S$.} distance between the program's output distributions on close inputs.

We now generalize the notion of stability to programs.

\begin{definition}[Output-Stable Programs]
    A program $P: \mathcal{X}\times\mathcal{E}\to\mathcal{Y}\times\mathcal{E}$ is $(\din \mapsto \dout)$-\emph{output stable} with respect to input metric $d_\mathcal{X}$ and output metric $d_\mathcal{Y}$ if $\forall x, x' \in \mathcal{X}$ such that $d_\mathcal{X}(x, x')\le \din$, and all pairs of input-compatible $\env, \env' \in \mathcal{E}$, there exists a coupling $(\Tilde{y}, \Tilde{y}')$ of the random variables $(\out(P(x, \env)), \out(P(x', \env')))$ such that $d_\mathcal{Y}(\Tilde{y}, \Tilde{y}')\le \dout$ with probability $1$.
\end{definition}

As an example, consider an output-pure program $P: \mathcal{X}\times\mathcal{E}\to\mathcal{Y}\times\mathcal{E}$ that takes an input dataset $x \in\{0, 1\}^n$ and outputs the sum $y = \sum_{i = 1}^n x_i$. Such a program would be $(1\mapsto 1)$-\emph{output stable} under the input distance metric $\dham$ and output distance metric $d_{\mathbb{N}}(y, y') = |y - y'|$, since changing a row's value in the input can affect the output sum by at most $1$. The program $P$ can therefore be made differentially private (Definition~\ref{def:dp-programs}) by chaining it with another program that adds Laplace noise with scale $1/\varepsilon$ to the output (see Section~\ref{sec:chaining} for a formal definition of program chaining and Lemma~\ref{lemma:noisy-sum-output-private} for a RAM program that computes a DP sum).

We now introduce \emph{privacy measures} which are arbitrary distances between probability distributions.

\begin{definition}[Generalized Privacy Measures]
    A \emph{privacy measure} is a tuple $(\mathcal{M}, \le, M)$ where $(\mathcal{M}, \le)$ is a partially-ordered set, and $M$ is a mapping of two random variables $X$ and $X'$ over the same measurable space to an element $M(X, X') \in \mathcal{M}$ satisfying:
    \begin{enumerate}
        \item (\emph{Post-processing}). For every function $g$, $M(g(X), g(X')) \le M(X, X')$
        \item (\emph{Joint convexity}). For any collection of random variables $(X_i, X'_i)_{i \in\mathcal{I}}$ and a random variable $I\sim \mathcal{I}$, if $M(X_i, X'_i)\le c$ for all $i$, then $M(X_I, X'_I) \le c$.
    \end{enumerate}
\end{definition}

We will frequently refer to a privacy measure only by the mapping $M$, where $\mathcal{M}$ and $\le$ are implicit. Some useful examples of privacy measures are \emph{pure} and \emph{approximate} DP which use the \emph{max-divergence} and \emph{smoothed max-divergence} as their respective distance mappings.

\begin{definition}[Max-Divergence]\label{def:max-divergence}
    The \emph{max-divergence} between two random variables $Y$ and $Z$ taking values from the same domain is defined to be:
    \[
        D_{\infty}(Y||Z) = \max_{S \subseteq \textrm{\supp}(Y)} \left[\ln \dfrac{\Pr[Y \in S]}{\Pr[Z \in S]} \right]
    \]
\end{definition}

\begin{lemma}
    A program $P: \mathcal{X}\times\mathcal{E}\to\mathcal{Y}\times\mathcal{E}$ is $\varepsilon$-differentially private if and only if for every pair of adjacent datasets $x$ and $x'$, and every pair of input-compatible execution environments $\env, \env'\in\mathcal{E}$, $D_{\infty}(\out(P(x, \env))||\out(P(x', \env'))) \le \varepsilon$ and $D_{\infty}(\out(P(x', \env'))||\out(P(x, \env))) \le \varepsilon$.
\end{lemma}

\begin{definition}[Smoothed Max-Divergence]\label{def:smooth-max-divergence}
    The \emph{smoothed max-divergence} between $Y$ and $Z$ is defined to be:
    \[
        D^{\delta}_{\infty}(Y||Z) = \max_{S \subseteq \supp(Y): \Pr[Y \in S] \geq \delta} \left[\ln \dfrac{\Pr[Y \in S]-\delta}{\Pr[Z \in S]} \right]
    \]
\end{definition}

\begin{lemma}\label{lemma:smoothed-divergence-dp}
    A program $P :\mathcal{X}\times\mathcal{E}\to\mathcal{Y}\times\mathcal{E}$ is $(\epsilon, \delta)$-differentially private if and only if for every pair of adjacent datasets $x$ and $x'$, and every pair of input-compatible execution environments $\env, \env' \in\mathcal{E}$, 
    \[
    D_{\infty}^\delta(\out(P(x, \env))||\out(P(x', \env'))) \leq \varepsilon
    \] and 
    \[
    D_{\infty}^\delta(\out(P(x', \env'))||\out(P(x, \env))) \leq \varepsilon
    \]
\end{lemma}

Other examples of privacy measures that can be used in our framework are concentrated DP~\cite{dwork2016concentrated,bun2016concentrated}, R\'enyi DP~\cite{mironov2017renyi}, and $f$-DP~\cite{dong2022gaussian}. We give a general definition of differential privacy for arbitrary input metrics and privacy measures. 

\begin{definition}\label{def:dp-programs}
    A program $P: \mathcal{X} \times \mathcal{E} \to \mathcal{Y} \times \mathcal{E}$ is $(\din \mapsto \dout)$-differentially private with respect to input metric $d_\mathcal{X}$ and privacy measure $M$ if for every pair of datasets $x, x' \in \mathcal{X}$ satisfying $d_\mathcal{X}(x, x')\le \din$ and every pair of input-compatible execution environments $\env, \env' \in \mathcal{E}$, $M(\out(P(x, \env)), \out(P(x', \env'))) \le \dout$.
\end{definition}

Definition~\ref{def:approx-dp-programs} is a special case of Definition~\ref{def:dp-programs} by replacing setting the privacy measure $M$ to be the smoothed max-divergence (Definition~\ref{def:smooth-max-divergence}), $\dout = (\varepsilon, \delta)$, $\din = 1$, and $d_\mathcal{X}$ is a dataset distance metric such as $\did$ or $\dham$.

Finally, we remark that $(\varepsilon,\delta)$-differentially private algorithms have the following properties.

\begin{lemma}[Post-processing \cite{dwork2006calibrating}]\label{lem:dp postprocess}
Let $M: \mathcal{X} \rightarrow \mathcal{Y}$ be an $(\varepsilon,\delta)$-differentially private function and $f: \mathcal{Y} \rightarrow \mathcal{Z}$ be a (possibly) randomized function.
Then $f\circ M: \mathcal{X} \rightarrow \mathcal{Z}$ is $(\varepsilon, \delta)$-differentially private.
\end{lemma}

\begin{lemma}[Composition \cite{dwork2006our}]\label{lem:dp composition}
Let $M_1: \mathcal{X} \rightarrow \mathcal{Y}_1$ be a $(\varepsilon_1,\delta_1)$-differentially-private function and $M_2: \mathcal{X} \rightarrow \mathcal{Y}_2$ be $(\varepsilon_2,\delta_2)$-differentially-private function. Then $(M_1\otimes M_2)(x) = (M_1(x), M_2(x))$ is a $(\varepsilon_1 + \varepsilon_2, \delta_1 + \delta_2)$-differentially-private function.
\end{lemma}

\section{Timing-Stable Programs}
\label{sec:rtstability}
In this section, we introduce new stability definitions for program runtime. Analogous to how global sensitivity determines an amount of noise that suffices for achieving differential privacy in randomized functions, we present a similar notion that determines an added runtime delay that suffices for ensuring timing privacy in randomized programs.

\subsection{Timing Stability}

\begin{definition}[Timing Stability]
Let $P: \mathcal{X} \rightarrow \mathcal{Y}$ be a (possibly randomized) program and $d_{\mathcal{X}}$ a metric on $\mathcal{X}$. Then we say that $P$ is $(\din\mapsto\tout)$-\emph{timing-stable} with respect to $d_{\mathcal{X}}$ if $\forall x, x'\in \mathcal{X}$ satisfying $d_{\mathcal{X}}(x, x') \le \din$, and all pairs of input-compatible $\mathtt{env}, \mathtt{env}' \in \mathcal{E}$, $\exists$ a coupling $(\Tilde{r}, \Tilde{r}')$ of $\TP{x, \mathtt{env}}$ and $\TP{x', \mathtt{env}'}$ such that $|\Tilde{r} -  \Tilde{r}'| \le \tout$ with probability $1$.
\end{definition}

Programs that exhibit timing stability ensure that changes in their inputs only cause bounded changes in their runtime distributions. However, this property is less useful for addressing privacy concerns when the program's output is also made available. For instance, consider the execution tree of a program that returns a randomized response to a Boolean input, as shown in Figure~\ref{fig:randresp}. Each path in the tree signifies a possible execution trace of the program on input $x$. An analysis of the tree reveals that the runtime does not depend on the input bit $x$. Programs with such input-independent runtime are categorized as $0$-timing-stable (see Appendix, Lemma~\ref{lemma:input-independence} for proof). Furthermore, the program outputs an independent Bernoulli random variable with $p = 1/2$ and therefore achieves $0$-DP. This might suggest that the program cannot leak any information about the input. However, the runtime exposes the program's internal randomness, indicating whether the output is $x$ or $1 - x$. Thus, the input bit can be precisely determined by the combination of the program's output and runtime, despite the program being both $0$-timing stable and $0$-DP. This demonstrates the need for a refined definition of timing stability that considers the program's output, which we provide in the next section.

\begin{figure}[t]
    \begin{minipage}{.5\textwidth}
        \begin{algorithm}[H]
            \vspace{5px}
            \raggedright\textbf{Input:} A bit $x \in \{0, 1\}$ at memory location $M[\texttt{input\_ptr}]$\\
            \raggedright\textbf{Output:} A uniform random bit \\
            \vspace{5px}
            \begin{algorithmic}[1]
                \STATE $\texttt{x} = M[\texttt{input\_ptr}]$;
                \STATE $\texttt{output\_len} = 1$;
                \STATE $\texttt{output\_ptr} = \texttt{input\_ptr}$;
                \STATE $\texttt{b} = $ \texttt{RAND}$(1)$; \COMMENT{flip coin}
                \IF{$\texttt{b} == 1$}
                \STATE $M[\texttt{output\_ptr}] = 1 - \texttt{x}$; \COMMENT{flip bit}
                \ENDIF
                \STATE $\texttt{HALT}$;
            \end{algorithmic}
            \caption{\\ A RAM program for randomized response.}
            \label{alg:ram-rr}
        \end{algorithm}
    \end{minipage}
    \begin{minipage}{.5\textwidth}
        \centering
        \begin{tikzpicture}[sibling distance=50mm, level distance=30mm, ->]
        \tikzset{every node/.style={align=left}, every edge/.style={draw, ->, thick}}
        \node {(1) read input $x$ \\ (2) set output length \\ (3) set output ptr \\ (4) sample a uniform bit $b$ \\ (5) conditional check on $b$ }
            child {
                node {(6) Halt}
                edge from parent node[left] {$b = 0$}
            }
            child {
                node {(6) Write $1 - x$ to output \\ (7) Halt }
                edge from parent node[right] {$b = 1$}
            };
        \end{tikzpicture}
    \end{minipage}%
    \caption{A Boolean randomized response RAM program and its corresponding execution tree. The program executes $1$ additional instruction when it outputs $1 - x$ versus $x$.}
    \label{fig:randresp}
\end{figure}

\subsection{Output-Conditional Timing Stability}

\begin{definition}[Output-Conditional Timing Stability]\label{def:output-conditional-timing-stability}
Let $P: \mathcal{X}\times\mathcal{E} \to \mathcal{Y}\times\mathcal{E}$ be (possibly randomized) program and $d_{\mathcal{X}}$ a metric on $\mathcal{X}$. Then we say that $P$ is $(\din\mapsto\tout)$-\emph{OC timing-stable} with respect to $d_{\mathcal{X}}$ if $\forall x,x'$ satisfying $d_{\mathcal{X}}(x, x') \le \din$, $\forall$ input-compatible environments $\env, \env' \in \mathcal{E}$, and $\forall y \in \supp(\out(P(x, \env))) \cap \supp(\out(P(x', \env')))$, there exists a coupling $(\Tilde{r}, \Tilde{r}')$ of the conditional runtime random variables $\TP{x, \env}|_{\out(P(x, \env)) = y}$ and $\TP{x', \env'}|_{\out(P(x', \env')) = y}$ such that $|\Tilde{r} - \Tilde{r}'| \le t_{out}$ with probability $1$.
\end{definition}

OC timing stability imposes stricter requirements than basic timing stability when the support sets of the random variables $\out(P(x, \env))$ and $\out(P(x', \env'))$ are similar. For any output $y$ that has a non-zero probability of being produced by $P$ for inputs $x$ and $x'$, as is implied if $P$'s output is differentially private (e.g., for both pure and Rényi DP, the supports must be identical), output-conditional timing stability ensures that the distributions $\TP{x, \env}|_{\out(P(x, \env)) = y}$ and $\TP{x', \env'}|_{\out(P(x', \env')) = y}$ are close. 

We can now analyze the randomized response program depicted in Figure~\ref{fig:randresp} through the perspective of OC timing stability, and show that the program is $(1 \mapsto 1)$-OC timing stable with respect to the Hamming distance metric $\dham$ (Lemma~\ref{lemma:rr-oc-stable}).

\begin{lemma}\label{lemma:rr-oc-stable}
    The Boolean randomized response program $P$ shown in Figure~\ref{fig:randresp} is $(1\mapsto 1)$-\emph{OC timing stable} with respect to the Hamming distance metric $\dham$. 
\end{lemma}

\begin{proof}
    $\TP{0, \env_0}$ is always $6$ conditioned on $\out(P(0, \env_0)) = 0$, and $\TP{1, \env_1}$ is always $7$ conditioned on $\out(P(1, \env_1)) = 0$. Thus, a point coupling of the conditional random variables $\TP{0, \env}|_{\out(P(0, \env_0)) = 0}$ and $\TP{1, \env_1}|_{\out(P(1, \env_1))=0}$ chosen as $(\Tilde{r}, \Tilde{r}') = (6, 7)$ satisfies $\Pr[|\Tilde{r} - \Tilde{r}'| \le 1] = 1$. 
    
    Similarly, we can take the point coupling of the random variables $\TP{0, \env_0}|_{\out(P(0, \env_0)) = 1}$ and $ \TP{1, \env_1}|_{\out(P(1, \env_1))=1}$ as $(\Tilde{r}, \Tilde{r}') = (7, 6)$ which also satisfies $\Pr[|\Tilde{r} - \Tilde{r}'| \le 1] = 1$.
\end{proof}

Constant-time programs are $(\din\mapsto 0)$-OC timing stable for all $\din$ (Lemma~\ref{lemma:constant-implies-oc}, Appendix). Additionally, any program that has deterministic output and is $(\din\mapsto\tout)$-timing-stable is also $(\din\mapsto\tout)$-OC timing stable (Lemma~\ref{lemma:rt-implies-oc}, Appendix).

\subsection{Jointly-Output/Timing Stable Programs}
We consider another natural definition of runtime stability, called \emph{joint output/timing stability}, which considers the joint distribution of a program’s output and its runtime. This notion becomes particularly useful when chaining programs together to create more complex functionality (\S\ref{sec:chaining}). In particular, joint output/timing stability ensures that close inputs simultaneously produce close outputs and runtimes. This property becomes useful when reasoning about the output-conditional timing stability of programs whose inputs come from the output of another program (\S\ref{sec:chaining}). In particular, an output-conditional timing-stable program $P_2$, bounds the change in runtime conditioned on a given output value for all $\din$-close inputs. However, when constructing a chained program $P_2 \circ P_1$, where $P_2$ receives its input from the output of $P_1$, it's necessary to jointly bound the output and runtime stability of $P_1$ to guarantee output-conditional timing stability for the chained program \(P_2\circ P_1\). 

\begin{definition}[Joint Output/Timing Stability]\label{def:joint-output-timing-stability} Let $P: \mathcal{X}\times\mathcal{E} \to \mathcal{Y}\times\mathcal{E}$ be a program and $d_{\mathcal{X}}$ a metric on $\mathcal{X}$ and $d_\mathcal{Y}$ a metric on $\mathcal{Y}$. Then we say that $P$ is $(\din \to\{\dout, \tout\})$-\emph{jointly output/timing stable} with respect to $d_{\mathcal{X}}$ and $d_\mathcal{Y}$ if $\forall x, x' \in \mathcal{X}$ satisfying $d_{\mathcal{X}}(x, x') \le \din$, and all pairs of input-compatible $\env, \env'\in\mathcal{E}$, $\exists$ a coupling $((\Tilde{u}, \Tilde{v}), (\Tilde{u}', \Tilde{v}'))$ of 
\begin{align*}
    ((\out(P(x, \env)), \TP{x, \env}), (\out(P(x', \env')), \TP{x', \env'}))
\end{align*} such that $d_\mathcal{Y}(\Tilde{u}, \Tilde{u}')\le \dout$ and $|\Tilde{v} - \Tilde{v}'| \le \tout$ with probability $1$.
\end{definition}

For programs with deterministic\footnote{A program $P: \mathcal{X}\times\mathcal{E}\to\mathcal{Y}\times\mathcal{E}$ is \emph{deterministic in its output} if for all $x \in \mathcal{X}$, $\exists y \in \mathcal{Y}$, such that for all input-compatible $\env \in \mathcal{E}$, $\Pr[\out(P(x, \env)) = y] = 1$. Therefore, deterministic programs are also output-pure programs.} outputs, if the program simultaneously satisfies ${(d_1\mapsto d_2)}$-output stability and $(d_1\mapsto t_1)$-timing stability, then the program will also satisfy $(d_1\mapsto \{d_2, t_1\})$-joint output/timing stability.
\shortpage
\shortpage

\begin{lemma}\label{lemma:det-joint-stability}
    If $P :\mathcal{X}\times\mathcal{E} \to \mathcal{Y}\times\mathcal{E}$ is a deterministic (in its output) program that is $(d_1 \mapsto d_2)$ output stable and $(d_1 \mapsto t_1)$-timing stable, then $P$ is $(d_1 \mapsto\{d_2, t_1\})$-jointly output/timing stable.
\end{lemma}

\begin{proof}
    Let $y =\out(P(x, \env))$ and $y' = \out(P(x', \env'))$ and choose the coupling to be $((y, r), (y', r'))$ where $(r, r')$ is sampled from the coupling $(\Tilde{r}, \Tilde{r}')$ associated with the timing stability of $P$. Then $d_\mathcal{Y}(y, y') \le d_2$ by output stability and $|r - r'| \le t_1$ by timing stability. Since $P$ is deterministic, it follows that $(y, r)$ and $(y', r')$ are identically distributed to $(\out(P(x, \env)), \TP{x, \env})$ and $(\out(P(x', \env')), \TP{x', \env'})$ respectively and the claim is satisfied.
\end{proof}

A useful example of a jointly output/timing stable RAM program is one that sums over the dataset.

\begin{algorithm}[t]
\vspace{5px}
\begin{flushleft}
\textbf{Input:} A dataset $x \in \{0,\ldots,\Delta\}^n$ located at $M[\texttt{input\_ptr}],\ldots,M[\texttt{input\_ptr} + \texttt{input\_len} - 1]$. We require that $\Delta < 2^\omega$ and $n \le 2^\omega - 1$ in the $\omega$-bit Word RAM model. 

\vspace{5px}
\textbf{Output:} $\sum M[i]$ for $i = \texttt{input\_ptr}, \dots ,(\texttt{input\_ptr} + \texttt{input\_len} - 1)$. The program outputs $\min\{\sum M[i], 2^\omega - 1\}$ in the Word RAM model.
\end{flushleft}
\vspace{5px}
\begin{algorithmic}[1]
\STATE $\texttt{output\_len} = 1$;
\STATE $\texttt{idx} = \texttt{input\_ptr}$;
\STATE $\texttt{n} = \texttt{input\_ptr} + \texttt{input\_len}$;
\STATE $\texttt{sum} = 0$;
\WHILE{$\texttt{idx} < \texttt{n}$}
\STATE $\texttt{sum} = M[\texttt{idx}] + \texttt{sum}$;
\STATE $\texttt{idx} = \texttt{idx} + 1;$
\ENDWHILE
\STATE $\texttt{output\_ptr} = 0$;
\STATE $M[\texttt{output\_ptr}] = \texttt{sum}$;
\STATE $\texttt{HALT}$;
\end{algorithmic}

\caption{Sum Program}\label{program:sum}
\end{algorithm}

\begin{lemma}\label{lemma:joint-stability-sum}
    The Sum Word RAM program $P:\mathcal{X}\times\mathcal{E}\to\mathbb{N}\times\mathcal{E}$ (Program~\ref{program:sum}) is $(1\mapsto \{\Delta, 3\})$-jointly output/timing stable under the insert-delete input distance metric $\did$ and the output distance metric $d_\mathbb{N}$ defined as $d_\mathbb{N}(y, y') = |y - y'|$.
\end{lemma}

\begin{proof}
Inserting or deleting an input record changes the runtime by one loop iteration, which consists of 3 instructions. Therefore the program is $(1\mapsto 3)$-timing stable under the input distance metric $\did$. Additionally, the program is $(1\mapsto \Delta)$-output stable under $\did$ since adding or removing an input record changes the sum by at most $\Delta$ (the maximum value of an input record). Since the program is deterministic in its output, the program satisfies $(1\mapsto \{\Delta, 3\})$-joint output/timing stability (Lemma~\ref{lemma:det-joint-stability}). The proof follows similarly if $P$ is instead a RAM program.
\end{proof}

We now show that the class of timing-stable programs is restricted to programs that have a worst-case runtime that is at most linear in their input length. This implies that certain programs, e.g., those that perform superlinear-time sorting, cannot achieve timing stability on inputs of unbounded length. It is an interesting challenge for future work to see if non-trivial timing-private programs can be designed that incorporate such operations.

\begin{lemma}[Timing-stable programs have $O(n)$ runtime]\label{lemma:timing-stable-implies-linear}
    Let $\mathcal{X} = \mathcal{D}^*$ for a row-domain $\mathcal{D}$. A program $P:\mathcal{X}\times\mathcal{E}\to\mathcal{Y}\times\mathcal{E}$ on inputs of unbounded length that is $(\din\mapsto\tout)$-timing-stable with respect to distance metric $d_\mathcal{X}(x, x') = ||x| - |x'||$ has at most linear runtime. Specifically, for all $x$ and all $t_0\in \mathcal{T}$, we have
    $$\Pr[\TP{x} \le |x| \cdot \tout + t_0]\ge \Pr[\TP{\lambda} \le t_0]$$ where $\lambda$ is the empty dataset of size $0$.
\end{lemma}

Note that for every $p < 1$, there exists a $t_0$ such that $\Pr[\TP{\lambda} \le t_0] > p$. Therefore the above lemma says that, with high probability, the runtime is at most linear.

\begin{proof}
     Consider a sequence of datasets $(x_0, x_1, \dots, x_n)$ such that $d_\mathcal{X}(x_i, x_{i + 1}) \le \din$ and $x_0 = \lambda$. Then for all $x_i, x_{i+1}$, and all pairs of input-compatible execution environments $\env_i, \env_{i+1} \in \mathcal{E}$, there exists a coupling $(\Tilde{r}_i, \Tilde{r}_{i+1})$ of $\TP{x_i, \env_i}$ and $\TP{x_{i+1}, \env_{i+1}}$ such that $|\Tilde{r}_i - \Tilde{r}_{i+1}| \le\tout$ by the $(\din\mapsto\tout)$-timing-stability of $P$. Using the Gluing Lemma for couplings~\cite{villani2009optimal}, we can construct a big coupling $(\Tilde{r}_0, \Tilde{r}_1, \dots, \Tilde{r}_n)$ such that $\Pr[|\Tilde{r}_i - \Tilde{r}_{i+1}| \le \tout] = 1$ for all $i$. Then, by an application of the triangle inequality, for all pairs of execution environments $\env_0, \env_n \in\mathcal{E}$, there exists a coupling $(\Tilde{r}_0, \Tilde{r}_n)$ of $\TP{x_0, \env_0}$ and $\TP{x_n, \env_n}$ such that $|\Tilde{r}_0 - \Tilde{r}_n| \le n\cdot \tout$ with probability $1$. Therefore 
     \begin{align*}
         &\Pr[\TP{x_n, \env_n} \le |x_n| \cdot \tout + t_0] \\ &= \Pr[\TP{x_n, \env_n} \le n\cdot \tout + t_0] \\
         &\ge \Pr[\Tilde{r}_0 \le t_0] \\
         &= \Pr[\TP{\lambda, \env_0} \le t_0]
     \end{align*}
\end{proof}

The result also holds for jointly output/timing-stable programs since jointly output/timing-stable programs are also timing-stable programs (Lemma~\ref{lemma:joint-implies-timing-stable}). However, a similar statement for OC-timing-stable programs does not hold, since it does not guarantee much about distributions of the runtimes on adjacent inputs when the output distributions on those inputs are very different.  (In an extreme case, when the output distributions have disjoint supports, OC-timing stability says nothing).
\section{Timing-Private Programs}
\label{sec:rtprivacy}
In this section, we introduce a notion of privacy with respect to timing attacks. Intuitively, we require that \emph{timing-private} programs should not leak much more information about their input than what is already revealed by their output distributions. 

\subsection{Output-Conditional Timing Privacy}
We generalize Definition~\ref{def:rt-privacy-simple} to work with arbitrary dataset distance metrics and privacy measures. 

\begin{definition}[Output-Conditional Timing Privacy]
\label{def:rtprivacy}
Let $P: \mathcal{X}\times\mathcal{E} \rightarrow \mathcal{Y}\times\mathcal{E}$ be a program, $d_{\mathcal{X}}$ a metric on $\mathcal{X}$, and $\mathcal{M}$ a distance measure between probability distributions. Then we say that $P$ is $(\din\mapsto \dout)$-\emph{OC-timing-private} with respect to $d_{\mathcal{X}}$ and $\mathcal{M}$ if for all $x, x'$ satisfying $d_{\mathcal{X}}(x, x') \le \din$, all pairs of environments $\env, \env' \in \mathcal{E}$, and all $y \in \supp(\out(P(x, \env))) \cap \supp(\out(P(x', \env')))$
\[\mathcal{M}(\TP{x, \env}|_{\out(P(x, \env)) = y}, \TP{x', \env'}|_{\out(P(x', \env')) = y}) \le \dout
\]
\end{definition}

We also provide an alternative and equivalent simulation-based definition of timing privacy. 
\begin{definition}[Sim-Timing Privacy]
\label{def:sim-rtprivacy}
 Let $P: \mathcal{X}\times\mathcal{E} \rightarrow \mathcal{Y}\times\mathcal{E}$ be a program, $d_{\mathcal{X}}$ a metric on $\mathcal{X}$, and $\mathcal{M}$ a distance measure between probability distributions. Then we say that $P$ is $(\din\mapsto\dout)$-timing private with respect to $d_{\mathcal{X}}$ and $\mathcal{M}$ if $\forall x, x'$ satisfying $d_{\mathcal{X}}(x, x') \le \din$, all environments $\env, \env' \in \mathcal{E}$, $\exists$ a simulator $S : \{(x, \env), (x', \env')\}\times \mathcal{Y}  \to \mathcal{T}$ such that:

    \begin{enumerate}
    \item $\forall y\in \supp(\out(P(x,\env))),$ we have $S(x, \env, y)\equiv \TP{x, \env}|_{\out(P(x, \env))=y}$ 
    \item $\forall y \in \supp(\out(P(x',\env')))$, we have $S(x', \env', y)\equiv \TP{x', \env'}|_{\out(P(x', \env'))=y}$ 
    \item $\forall y \in \mathcal{Y}, \mathcal{M}(S(x, \env, y), S(x', \env', y)) \le \dout$
    \end{enumerate}
\end{definition}

where $Y \equiv Z$ denotes that the random variables $Y$ and $Z$ are identically distributed. \\

The alternative definition is equivalent to Definition~\ref{def:rtprivacy}, but it lends itself more naturally to a computational analogue of timing privacy when we add the assumption that $S$ runs in polynomial time. Requirements (1) and (2) that $S$ is identically distributed to the conditional runtime random variable gets relaxed to computational indistinguishability, and requirement (3) gets relaxed to computational differential privacy~\cite{mironov2009computational}. Importantly, such a definition is meaningful even if $\supp(\out(P(x, \env)))\cap \supp(\out(P(x', \env'))) = \emptyset$, in contrast to Definition~\ref{def:rtprivacy}. The supports on adjacent inputs may be disjoint, but a polynomial-time simulator $S$ will not be able to detect this. 

\begin{lemma}\label{lemma:timing-sim-equiv}
Timing privacy (Definition~\ref{def:rtprivacy}) and Sim-Timing Privacy (Definition~\ref{def:sim-rtprivacy}) are equivalent.
\end{lemma}
\begin{proof}
    We first show that Definition~\ref{def:rtprivacy} $\implies$ Definition~\ref{def:sim-rtprivacy}. On every pair of inputs $x, x'$ satisfying $d_\mathcal{X}(x, x')\le \din$, all pairs of input-compatible environments $\env, \env' \in\mathcal{E}$, and every output $y \in \supp(Y) \cap \supp(Y')$, where $Y = \out(P(x, \env))$ and $Y' = \out(P(x', \env'))$, we have that 
    \begin{align*}
        \mathcal{M}(\TP{x, \env}|_{Y=y}, \TP{x', \env'}|_{Y'=y}) \le \dout
    \end{align*}
    
    by timing privacy. Thus, let 
    \begin{align*}
        S(x, \env, y) =
        \begin{cases}
            t\sim \TP{x, \env}|_{Y = y} & \text{if } y \in \supp(Y) \\
            t\sim \TP{x', \env'}|_{Y' = y} & \text{if } y \in \supp(Y') \setminus \supp(Y) \\
            0 & \text{otherwise}
        \end{cases} 
    \end{align*} 
     defined similarly for $S(x', \env', y)$. We now make the following claims.
    \begin{enumerate}[label=(\arabic*)]
        \item $\forall y \in \supp(Y)$
        \begin{align*}
            S(x, \env, y) \equiv \TP{x, \env}|_{Y = y}
        \end{align*}
        \item $\forall y \in \supp(Y')$
        \begin{align*}
            S(x', \env', y) \equiv \TP{x', \env'}|_{Y' = y}
        \end{align*}
        \item $\forall y\in\mathcal{Y}, \mathcal{M}(S(x, \env, y), S(x', \env', y)) \le \dout$
    \end{enumerate} 

    Claims (1) and (2) are trivially true by the definition of $S$. The proof of claim (3) follows from the fact that for all $y \not\in \supp(Y)\cup \supp(Y')$, the simulator outputs $0$ and therefore $S(x, \env, y)\equiv S(x', \env', y)$. Similarly, for all  $y \in \supp(Y)\cap \supp(Y')$, the simulator outputs either $t\sim\TP{x, \env}$ or $t\sim\TP{x', \env'}$ according to its input. For all such $y$, $M(\TP{x, \env}, \TP{x', \env'})\le \dout$ by timing privacy. The only remaining case is when $S(x, \env, y)$ is given a value $y \in \supp(Y') \setminus \supp(Y)$. In this scenario, $S$ outputs $t\sim \TP{x', \env'}$ which will be distributed identically to $S(x', \env', y)$ by definition (and $S(x', \env', y)$ behaves similarly for $y \in \supp(Y) \setminus \supp(Y')$). Therefore $S(x, \env, y) \equiv S(x', \env', y)$ for all such $y$ and the claim is proven. \\

    Now we show that Definition~\ref{def:sim-rtprivacy} $\implies$ Definition~\ref{def:rtprivacy}. For every pair of datasets $x, x'$ satisfying $d_\mathcal{X}(x, x')\le \din$, and for all input-compatible environments $\env, \env' \in \mathcal{E}$, there exists a simulator $S: \{(x, \env), (x', \env')\}\times\mathcal{Y}\to \mathcal{T}$ such that 
    \begin{enumerate}
    \item $\forall y\in \supp(Y), S(x, \env, y)\equiv \TP{x, \env}|_{Y=y}$ 
    \item $\forall y \in \supp(Y')$, $S(x', \env', y)\equiv \TP{x', \env'}|_{Y'=y}$ 
    \item $\forall y \in \mathcal{Y}, \mathcal{M}(S(x, \env, y), S(x', \env', y)) \le \dout$
    \end{enumerate}

    by sim-timing privacy. It follows that $P$ is $(\din\mapsto\tout)$-timing private since for all $y\in\supp(Y)$,  $S(x, \env, y)$ is identically distributed to the conditional runtime random variable $\TP{x, \env}|_{Y = y}$ (and similarly for $S(x', \env', y)\equiv \TP{x', \env'}|_{Y'=y}$ for all $y\in\supp(Y')$). Additionally, since
    \begin{align*}
        \forall y\in\mathcal{Y}, \mathcal{M}(S(x, \env, y), S(x', \env', y)) \le \dout
    \end{align*}

    it follows that
    \begin{align*}
        \mathcal{M}(\TP{x, \env}|_{Y = y}, \TP{x', \env'}|_{Y' = y}) \le \dout
    \end{align*}

    which is what we wanted to show.
\end{proof}

\subsection{Joint Output/Timing Privacy}\label{sec:joint-output-privacy}
We now give a more general notion of \emph{joint output/timing privacy} for arbitrary dataset metrics and privacy measures, generalizing the approximate DP version in Definition~\ref{def:approx-joint-privacy}, which is from~\cite{ben2023resistance}.
\begin{definition}[Joint Output/Timing Privacy]\label{def:joint-privacy}
    Let $P: \mathcal{X}\times\mathcal{E} \rightarrow \mathcal{Y}\times\mathcal{E}$ be a (possibly) randomized program. Then we say that $P$ is $(\din\mapsto\dout)$-\emph{jointly output/timing-private} with respect to distance metric $d_\mathcal{X}$ and privacy measure $M$, if for all adjacent $x, x' \in \mathcal{X}$ satisfying $d_{\mathcal{X}}(x, x') \le \din$, and all pairs of input-compatible execution environments $\env, \env' \in \mathcal{E}$
        \begin{align*}
            M(\out(P(x, \env)), \TP{x, \env}), (\out(P(x', \env')), \TP{x', \env'})) \le \dout
        \end{align*}   
\end{definition}
In contrast to timing privacy (Definition~\ref{def:rtprivacy}), \emph{joint output/timing privacy} applies the standard DP definition to the joint random variable containing the program's output and runtime. In the case of pure-DP, this is equivalent to the program being both $\Theta(\varepsilon)$-differentially private and $\Theta(\varepsilon)$-timing private (Lemma~\ref{lemma:pure-joint-implies-timing}). However, we remark that the result does not extend to $(\varepsilon, \delta)$-timing privacy. For example, a $(\varepsilon, \delta)$-jointly output/timing-private program may, with probability $\delta$, output a special constant $y^*$ and completely encode the input in its running time. On the other hand, such a program would not satisfy timing privacy, which would require that the runtime be DP even conditioned on this rare, special output $y^*$.

\begin{lemma}\label{lemma:pure-joint-implies-timing}
    If a program $P: \mathcal{X}\times\mathcal{E}\to \mathcal{Y}\times\mathcal{E}$ is $\varepsilon$-jointly output/timing private then it is also $\varepsilon$-DP and $2\varepsilon$-timing-private.
\end{lemma}
\begin{proof}
   Let $Y = \out(P(x, \env))$ and $Y' = \out(P(x', \env'))$. We start with the fact that 
    \begin{align*}
        & \Pr[(Y, \TP{x, \env}) \in \{y\}\times S_2] \\
        &\quad\quad= \Pr[\TP{x, \env} \in S_2 | Y = y ]\cdot\Pr[Y = y ] \\
        &\quad\quad\le e^\varepsilon\cdot\Pr[\TP{x', \env'} \in S_2 | Y' = y]\cdot\Pr[Y' = y] 
    \end{align*}

    by joint output/timing privacy. Note that if 
    \begin{align*}
        \Pr[(Y, \TP{x, \env}) \in \{y\}\times S_2] \le  e^\varepsilon \cdot\Pr[(Y', \TP{x', \env'}) \in \{y\}\times S_2] 
    \end{align*}
    then
    \begin{align*}
        \Pr[Y = y] \le e^\varepsilon \cdot\Pr[Y' = y]
    \end{align*}
    by post-processing (Lemma~\ref{lem:dp postprocess}). Thus, 
    \begin{align*}
        \Pr[\TP{x, \env} \in S_2 | Y = y] &= \frac{\Pr[(Y, \TP{x, \env}) \in \{y\}\times S_2]}{\Pr[Y = y] } \\ \\
        &\le \frac{e^\varepsilon\cdot \Pr[(Y', \TP{x', \env'}) \in \{y\}\times S_2]}{e^{-\varepsilon}\cdot \Pr[Y' = y] } \\
        &= e^{2\varepsilon} \cdot\Pr[\TP{x', \env'} \in S_2 | Y' \in S_1] \\
    \end{align*}
\end{proof}

\begin{lemma}\label{lemma:timing-privacy-output}
    If $P: \mathcal{X}\times\mathcal{E}\to \mathcal{Y}\times\mathcal{E}$ is $(\varepsilon_1, \delta_1)$-differentially private and $(\varepsilon_2, \delta_2)$-timing-private, then $P$ achieves $(\varepsilon_1 + \varepsilon_2, \delta_1 + \delta_2)$-joint output/timing privacy.
\end{lemma}

\begin{proof}
    The proof follows as an application of the composition theorem (Lemma~\ref{lem:dp composition}).  Since the program is $(\varepsilon_1, \delta_1)$-DP, we have that for any pair of $\din$-close inputs $x, x'$, all input-compatible $\env, \env'\in\mathcal{E}$, and for all $S \subseteq \mathcal{Y}$ 
    \begin{align*}
        \Pr[Y \in S] \le e^{\varepsilon_1} \cdot \Pr[Y' \in S] + \delta_1
    \end{align*}
    where $Y = \out(P(x, \env))$ and $Y' = \out(P(x', \env'))$.
    
    By timing privacy, for any pair of $\din$-close inputs $x, x'$, and all input-compatible $\env, \env'\in\mathcal{E}$, there exists a simulator $S$ satisfying for all $y \in \supp(Y)$
    \begin{align*}
        S(x, \env, y) \equiv \TP{x, \env}|_{Y=y}
    \end{align*}
and similarly for and all $y \in \supp(Y')$,
    \begin{align*}
        S(x', \env', y) \equiv \TP{x', \env'}|_{Y'=y}
    \end{align*} 
    Finally, for all $y\in\mathcal{Y}$, and all $Q\subseteq \mathcal{T}$ we have that
    \begin{align*}
        \Pr[S(x, \env, y)\in Q] \le e^\varepsilon_2 \cdot \Pr[S(x', \env', y)\in Q] + \delta_2
    \end{align*}
    
    Thus, the random variable $(Y, S(x, \env, Y)) \equiv (Y, \TP{x, \env}|_{Y})$ is equivalent to the composition of an $(\varepsilon_1, \delta_1)$-DP and $(\varepsilon_2, \delta_2)$-DP mechanism.
\end{proof}

In contrast to timing privacy (Definition~\ref{def:rtprivacy}), \emph{joint output/timing privacy} applies the standard DP definition to the joint random variable containing the program's output and runtime. In the case of pure-DP, this is equivalent to the program being both $\Theta(\varepsilon)$-differentially private and $\Theta(\varepsilon)$-timing private (see Appendix, Lemma~\ref{lemma:pure-joint-implies-timing}). However, we remark that the result does not extend to $(\varepsilon, \delta)$-timing privacy. For example, a $(\varepsilon, \delta)$-jointly output/timing-private program may, with probability $\delta$, output a special constant $y^*$ and completely encode the input in its running time. On the other hand, such a program would not satisfy timing privacy, which would require that the runtime be DP even conditioned on this rare, special output $y^*$.

\section{Timing-Private Delay Programs}
We now discuss a core component for building timing private programs in our framework. We introduce \emph{timing-private delay programs} which operate similarly to additive noise mechanisms for output privacy. Such programs are used to delay the execution of a program before releasing its output. When applied to timing-stable programs, the addition of such delay can transform these programs into timing-private ones.

\subsection{Timing-Private Delays}\label{sub-sec:delays}
\begin{definition}[Timing-Private Delays]
A distribution $\Phi$ on a time domain $\mathcal{T} \subseteq \mathbb{R}_{\ge 0}$ satisfying $\Pr[\Phi < 0] = 0$ is a $(t_{in}\to d_{out})$-\emph{timing-private delay distribution} under privacy measure $\mathcal{M}$ if $\forall t \in \mathcal{T}$, $t\le t_{in}$, if $T\sim \Phi$, then
\begin{align*}
    \mathcal{M}(T, T + t) \le \dout 
\end{align*}
\end{definition}

\begin{remark}
The requirement that $\Phi$ has zero probability mass on values less than $0$ enforces the physical reality that time cannot be subtracted. 
\end{remark}

\noindent We now give an example of a distribution $\Phi$ that is a $(t_{in} \to (\varepsilon, \delta))$-\emph{timing-private delay distribution} under the smoothed max-divergence privacy measure $D_{\infty}^\delta$ (Definition~\ref{def:smooth-max-divergence}).

Recall the Discrete Laplace distribution with shift $\mu$ and scale $s$ has a probability density function:
\[
f(x | \mu, s) = \frac{e^{1/s} - 1}{e^{1/s} + 1}\cdot e^{-|x - \mu|/s}
\]

and CDF:
\[
F(x | \mu, s) = 
\begin{cases} 
\frac{e^{1/s}}{e^{1/s} + 1}\cdot e^{-(\mu - x)/s}, & \text{if } x \leq \mu \\
1 - \frac{1}{e^{1/s} + 1}\cdot e^{-(x - \mu)/s}, & \text{otherwise }
\end{cases}
\]

\begin{lemma}\label{lemma:disclap-delay-dist}
    For $\mu \ge t_{in}$, $B \ge 2\mu$, $s = t_{in}/\varepsilon$ with $t_{in}, \varepsilon > 0$, the \emph{censored} Discrete Laplace distribution $\Phi = \min \{\max \{T, 0\}, B\} + c$, for a constant $c$ where $T$ is sampled from a Discrete Laplace distribution with parameters $\mu$ and $s$, is a $(t_{in} \to (\varepsilon, \delta))$-\emph{timing private delay distribution} under the smoothed max-divergence privacy measure $D_{\infty}^\delta$ with $\delta = 2\cdot e^{-\varepsilon(\mu - t_{in})/t_{in}}$.
\end{lemma} 
\begin{proof}
We take $t_1 \le t_{in}$ and let $\phi \sim \Phi$. For all $t_1 < t \le B + c$, we have that: 
\begin{align*}
\frac{\Pr[\phi + t_1 = t]}{\Pr[\phi = t]} = e^{\frac{\varepsilon|t - t_1 - \mu| - |t - \mu|}{t_{in}}} \le e^{\frac{\varepsilon|t_1|}{t_{in}}} \le e^{\varepsilon}
\end{align*}

Additionally, for all $t < c$ it follows that $\Pr[\phi + t_1 = t] = \Pr[\phi = t] = 0$. Similarly, for all $t > B + c + t_1$ it follows that $\Pr[\phi + t_1 = t] = \Pr[\phi = t] = 0$. However, for all $c \le t < t_1$ we have that $\Pr[\phi + t_1 = t] = 0$ and $\Pr[\phi = t] > 0$ and therefore the multiplicative distance between the distributions is unbounded on this interval. This event only happens when $\phi \le t_1$ yielding

\begin{align*}
    \delta = F(t_1|\mu, s) \le F(t_{in}|\mu, s) &= \frac{e^{\varepsilon/t_{in}}}{e^{\varepsilon/t_{in}} + 1}\cdot e^{-\varepsilon(\mu - t_{in})/t_{in}} \\ &\le e^{-\varepsilon(\mu - t_{in})/t_{in}}
\end{align*}

by the CDF of the Discrete Laplace distribution. Similarly, for $B + c < t$, we have that $\Pr[\phi + t_1 = t] > 0$ and $\Pr[\phi = t] = 0$. This event only happens when $\phi > B - t_1$:
\begin{align*}
    \Pr[\phi > B - t_1] &\le 1 - F(B - t_{in} |\mu, s) \\
    &= \frac{1}{e^{\varepsilon/t_{in}} + 1}\cdot \big(e^{-\varepsilon(B - t_{in} - \mu)/t_{in}} \big) \\
    &\le e^{-\varepsilon(\mu - t_{in})/t_{in}}
\end{align*}

It follows that  $D^{\delta}_{\infty}(\phi + t_1||\phi) \le \varepsilon$ and $D^{\delta}_{\infty}(\phi||\phi + t_1) \le \varepsilon$ for $\delta = 2e^{-\varepsilon(\mu - t_{in})/t_{in}}$. Finally, we note that $\Pr[\phi < 0] = 0$. 
\end{proof}

\begin{definition}[Timing-Private Delay Program]
A $(t_{in} \mapsto \dout)$-\emph{timing-private delay program} $P : \mathcal{X}\times\mathcal{E} \to \mathcal{X}\times\mathcal{E}$ with respect to time domain $\mathcal{T}\subseteq \mathbb{R}_{\geq 0}$ and probability distance measure $\mathcal{M}$ implements the following functionality:

\begin{enumerate}
    \item \textbf{Identity function}. For all $x \in \mathcal{X}$, and all $\env \in\mathcal{E}$
    \begin{align*}
        \Pr[\out(P(x, \env)) = x] = 1
    \end{align*}
    \item \textbf{Timing-private delay distributed runtime}. For all $x \in \mathcal{X}$, all $\env\in\mathcal{E}$, $\TP{x, \env}$ is distributed according to a $(t_{in}\mapsto\dout)$-timing-private delay distribution
    on $\mathcal{T}$.
\end{enumerate}
\end{definition}

\begin{lemma}\label{lemma:timing-private-delay-programs-are-output-pure}
    Let $P:\mathcal{X}\times\mathcal{E}\to\mathcal{X}\times\mathcal{E}$ be a timing-private delay program. Then $P$ is output-pure.
\end{lemma}
\begin{proof}
    Since $P$ implements the identity function by definition, for all $x\in \mathcal{X}$, the random variables $\out(P(x, \env))$ and $\out(P(x, \env'))$ are identically distributed for all pairs of input-compatible execution environments $\env, \env' \in \mathcal{E}$.
\end{proof}

Program~\ref{program:disclap-timing-private} is an example of a timing-private delay program in the RAM and Word RAM models. 
The program has the property that it runs in time identically distributed to a \emph{censored}  Discrete Laplace distribution. Thus, we have:

\begin{lemma}\label{lemma:disclap-timing-private-delay}
    The timing-private delay Word RAM program $P: \mathcal{X}\times\mathcal{E}\to\mathcal{X}\times\mathcal{E}$ (Program~\ref{program:disclap-timing-private}) with $\texttt{shift} = \mu$ and scale parameter $s = 1/\ln(\texttt{b}/\texttt{a})$ satisfying $\ln(\texttt{b}/\texttt{a}) = \varepsilon/t_{in}$ is
    a $(t_{in} \mapsto (\varepsilon,\delta))$-\emph{timing-private delay program} on time domain $\mathbb{N}$
    for $\delta = 2\cdot e^{-\varepsilon(\mu - t_{in})/t_{in}}$.
\end{lemma}

\begin{proof}
    We will show that Program~\ref{program:disclap-timing-private} runs in time distributed exactly to $\min\{\max \{0, T\}, B\} + c$ where $c = 16 + 7\cdot\texttt{bound}$ and $T$ is the Discrete Laplace distribution with parameter $\mu = \texttt{shift}$, $s = t_{in} / \varepsilon = 1/\ln(\texttt{b}/\texttt{a})$, and $B \ge 2\cdot \mu$. The claim then follows from Lemma~\ref{lemma:disclap-delay-dist} and the fact that the program implements the identity function.
    
    The program implements the following functionality. First, the program samples a uniform random bit that is assigned to the variable $\texttt{sign}$ (line 6). Next, it flips a biased coin with probability $p = (\texttt{b} - \texttt{a})/(\texttt{b} + \texttt{a})$ (lines 7-10) and sets $\texttt{sample} = 0$ if the trial is successful (line 11). 
    Lines 1-14 take exactly $12$ instructions to execute independent of the branching condition due to the $\texttt{NOP}$ instructions. The program then loops $\texttt{bound}$ times, where each loop takes exactly $7$ instructions to execute. During this loop, if $\texttt{set} = 0$ (which happens conditioned on $\texttt{sample} \ne 0$), the program sets $\texttt{sample}$ to be a value drawn from a \emph{censored} Geometric distribution with $p = \frac{\texttt{b} - \texttt{a}}{\texttt{b}}$ and a bound of $t = \texttt{bound}$, where the censored Geometric distribution has probability mass:
    
    \begin{align*}
        \text{CensGeo}(x| p, t) =
        \begin{cases}
        p\cdot (1 - p)^{x - 1} & \text{if } x < t \\
        (1-p)^{t-1} & x = t \\
        0 & \text{otherwise}
        \end{cases}
    \end{align*}

    for $x \ge 1$. The distribution is censored since the program exits the loop when $\texttt{count} = \texttt{bound}$ (line 15) regardless of whether or not a success has been observed. In this case, the program sets $\texttt{sample} = \texttt{bound}$ (lines 26-29), which always executes in exactly $2$ instructions.
    
     The program then delays execution by the value $\texttt{sleep} = \texttt{shift} \pm \texttt{sample}$ (lines 30 - 33) where $\texttt{sample}$ is subtracted or added to $\texttt{shift}$ depending on the value of $\texttt{sign}$. Lines 30-33 always execute exactly $2$ instructions. Therefore, up to the final $\texttt{NOP}(\texttt{sleep})$ instruction (line 34), the program executes in exactly $c = 12 + 7\cdot\texttt{bound} + 2 + 2 = 16 + 7\cdot\texttt{bound}$ instructions. What remains to be shown is that $\texttt{sleep}$ is distributed exactly to $\Phi = \min\{ \max \{0, T\}, B\}$.

    We start with the case that $\texttt{sleep} = \mu$. Let $Z$ be the random variable associated with the coin flip at lines 7-10. Then $\Pr[Z = 1] = \frac{\texttt{b} - \texttt{a}}{\texttt{b} + \texttt{a}}$. Observe that for the PMF $f$ of the Discrete Laplace distribution:
    \begin{align*}
        f(x=\texttt{shift}|\mu = \texttt{shift}, s= 1/\ln(\texttt{b}/\texttt{a})) &= \frac{e^{\ln(\texttt{b}/\texttt{a})} - 1}{e^{\ln(\texttt{b}/\texttt{a})} + 1} \\
        &= \frac{\texttt{b} - \texttt{a}}{\texttt{b} + \texttt{a}} \\
        &= \Pr[Z = 1]
    \end{align*}
    
    and therefore the program delays for $\texttt{shift} = \mu$ instructions with probability mass according to DiscreteLaplace($\mu, s$). 
    
    We now consider the case where $\texttt{sleep} \ne \mu$, i.e., $Z = 0$. In this case, the program samples another value from a \emph{censored} Geometric distribution with $p = \frac{\texttt{b} - \texttt{a}}{\texttt{b}}$ (lines 15-29).  Thus, for all values $0 < y <  B$, $y \ne \mu$
    \begin{align*}
        \Pr[\texttt{sleep} = y] &= \frac{1}{2}\cdot (1 - \Pr[Z = 1]) \cdot \Pr[\textrm{CensGeo}(p, t) = |y - \mu|] \\
        &= \frac{1}{2}\cdot \big(1 - \frac{\texttt{b} - \texttt{a}}{\texttt{b} + \texttt{a}}\big)\cdot \big(\frac{\texttt{a}}{\texttt{b}}\big)^{|y - \mu| - 1}\cdot (1 - \frac{\texttt{a}}{\texttt{b}}) \\
        &= \frac{1}{2}\cdot(1 - \frac{e^{1/s} - 1}{e^{1/s} + 1})\cdot (e^{-1/s})^{(|y - \mu| - 1)} \cdot (1 - e^{-1/s}) \\
        &= \frac{1}{2}\cdot (\frac{e^{1/s} + 1 - e^{1/s} + 1}{e^{1/s} + 1})\cdot (e^{-1/s})^{(|y - \mu| - 1)} \cdot (1 - e^{-1/s}) \\
        &= \frac{1}{e^{1/s} + 1}\cdot (e^{-1/s})^{(|y - \mu| - 1)} \cdot (1 - e^{-1/s}) \\
        &= \frac{e^{(-(|y - \mu|) + 1)/s}- e^{-(|y - \mu|)/s}}{e^{1/s} + 1} \\
        &= \frac{e^{1/s} - 1}{e^{1/s} + 1}\cdot e^{-|y - \mu|/s}
    \end{align*}
    which equals the PMF of the Discrete Laplace distribution for all $0 < y < B, y \ne \mu$. 
    
    We now consider the edge cases. We show that $\Pr[\texttt{sleep} = 0] = F(0|\mu, s)$ where $F$ is the CDF of the Discrete Laplace distribution with shift $\mu$ and with scale parameter $s=1/\ln(\texttt{b}/\texttt{a})$ satisfying $\ln(\texttt{b}/\texttt{a}) = \varepsilon/t_{in}$:

    \begin{align*}
        \Pr[\texttt{sleep} = 0] &= \frac{1}{2}\cdot (1-\Pr[Z = 1])\cdot \Pr[\textrm{CensGeo}(p, t) \ge \mu - 1] \\
        &= \frac{1}{2}\cdot \big(1 - \frac{\texttt{b} - \texttt{a}}{\texttt{b} + \texttt{a}}\big)\cdot \big(\frac{\texttt{a}}{\texttt{b}}\big)^{(\mu - 1)} \\
        &= \frac{1}{2}\cdot(1 - \frac{e^{\varepsilon/t_{in}} - 1}{e^{\varepsilon/t_{in}} + 1})\cdot (e^{-\varepsilon/t_{in}})^{(\mu - 1)} \\
         &= \frac{1}{e^{\varepsilon/t_{in}} + 1}\cdot (e^{-\varepsilon\cdot (\mu - 1)/t_{in}}) \\
        &= \frac{e^{\varepsilon/t_{in}}}{e^{\varepsilon/t_{in}} + 1}\cdot (e^{-\varepsilon\cdot \mu/t_{in}}) \\
         &= F(0 | \mu = \mu, s=\varepsilon/t_{in})
    \end{align*}
    where the leading $\frac{1}{2}$ comes from the probability that the sign is negative.
    
    Finally, we show that $\Pr[\texttt{sleep} = B] = 1 - F(B - 1|\mu, s)$ where $F$ is the CDF of the Discrete Laplace distribution with shift $\mu$ and scale parameter $s=1/\ln(\texttt{b}/\texttt{a})$ satisfying $\ln(\texttt{b}/\texttt{a}) = \varepsilon/t_{in}$ and $B = \mu + \texttt{bound} \ge 2\mu$:
    \begin{align*}
        \Pr[\texttt{sleep} = B] &= \frac{1}{2}\cdot (1-\Pr[Z = 1])\cdot \Pr[\textrm{CensGeo}(p, t) = B] \\
        &= \frac{1}{2}\cdot \big(1 - \frac{\texttt{b} - \texttt{a}}{\texttt{b} + \texttt{a}}\big)\cdot \big(\frac{\texttt{a}}{\texttt{b}}\big)^{B - 1} \\
        &= \frac{1}{2}\cdot(1 - \frac{e^{\varepsilon/t_{in}} - 1}{e^{\varepsilon/t_{in}} + 1})\cdot (e^{-\varepsilon/t_{in}})^{B - 1} \\
         &= \frac{1}{e^{\varepsilon/t_{in}} + 1}\cdot (e^{-\varepsilon\cdot (B - 1)/t_{in}}) \\
        &= \frac{1}{e^{\varepsilon/t_{in}} + 1}\cdot (e^{-\varepsilon\cdot (B - 1)/t_{in}}) \\
         &= 1 - F(B - 1 | \mu = \mu, s=\varepsilon/t_{in}) 
    \end{align*}

    Thus, the program has runtime that is distributed exactly to $\Phi = \min \{\max\{ 0, T\}, B\} + c$ for $c = 16 + 7\cdot\texttt{bound}$. The lemma follows from Lemma~\ref{lemma:disclap-delay-dist}.
\end{proof}

\begin{lemma}\label{lemma:disclap-timing-private-delay-pure}
    The timing-private delay Word RAM program $P: \mathcal{X}\times\mathcal{E}\to\mathcal{X}\times\mathcal{E}$ (Program~\ref{program:disclap-timing-private}) is
    both output-pure and timing-pure.
\end{lemma}

\begin{proof}
    The proof follows from the fact that the program never interacts with uninitialized memory. Therefore for all $y \in \mathbb{N}$,  the program's output and runtime is identically distributed on all pairs of input-compatible execution environments $\env, \env' \in\mathcal{E}$.
\end{proof}


\begin{algorithm}[H]
\vspace{5px}
\begin{flushleft}
    
\textbf{Input:} An input $x$ located in memory $M[\texttt{input\_ptr}],\ldots, M[\texttt{input\_ptr} + \texttt{input\_len} - 1]$. 

\vspace{5px}
\textbf{Output:} The input $x$,
with runtime $\min\{ \max\{T,0\}, \texttt{shift} + \texttt{bound}\} + c$, 
where
$T$ is drawn from a Discrete Laplace distribution with parameters $\mu=\texttt{shift}$ and $s=1/\ln(\texttt{b}/\texttt{a})$ where
$\texttt{shift}$, $\texttt{a}$, $\texttt{b}>\texttt{a}$, and $\texttt{bound} \ge \texttt{shift}$ are hardcoded constants, and $c = 16 + 7\cdot\texttt{bound}$. For the Word RAM version of the program, we require $s, (\mu + \texttt{bound})$ and $(\texttt{a} + \texttt{b})$ to be less than $2^\omega$.
\end{flushleft}

\begin{algorithmic}[1]
\STATE $\texttt{output\_ptr} = \texttt{input\_ptr}$;
\STATE $\texttt{output\_len} = \texttt{input\_len}$; \COMMENT{compute the identity function}
\STATE $\texttt{count} = 0$;
\STATE $\texttt{set} = 0$;
\STATE $\texttt{sample} = 0$;
\STATE $\texttt{sign} = \texttt{RAND}(1)$; \COMMENT{add or subtract from $\mu$}
\STATE \texttt{zeroprobA} = $\texttt{b} - \texttt{a}$; \COMMENT{num. of prob. to sample a zero}
\STATE \texttt{zeroprobB} = $\texttt{b} + \texttt{a}$; \COMMENT{denom. of prob. to sample a zero}
\STATE $\texttt{idx} = \texttt{RAND}(\texttt{zeroprobB} - 1)$;
\item[]
\IF{$\texttt{idx} < \texttt{zeroprobA}$}
    \STATE $\texttt{sample} = 0$; \COMMENT{With probability $\frac{\texttt{b} - \texttt{a}}{\texttt{b} + \texttt{a}}$ add $\mu$ delay}
    \STATE $\texttt{set} = 1$;
\ELSE
    \STATE $\texttt{NOP}(2)$;
\ENDIF
\item[]

\WHILE{$\texttt{count} < \texttt{bound}$}
    \STATE $\texttt{count} = \texttt{count} + 1$;
    \IF{$\texttt{set} == 0$}
        \STATE $\texttt{idx} = \texttt{RAND}(\texttt{b} - 1)$;
        \IF{$\texttt{idx} < \texttt{b} - \texttt{a}$}
            \STATE $\texttt{sample} = \texttt{count}$; \COMMENT{with probability $\frac{\texttt{b} - \texttt{a}}{\texttt{b}}$ add $\mu \pm \texttt{count}$ delay}
            \STATE $\texttt{set} = 1$;
        \ELSE
            \STATE $\texttt{NOP}(2)$;
        \ENDIF
    \ELSE
        \STATE $\texttt{NOP}(4)$;
    \ENDIF
\ENDWHILE
\item[]

\IF{$\texttt{set} == 0$}
\STATE $\texttt{sample} = \texttt{bound}$;
\ELSE
\STATE $\texttt{NOP}(1)$;
\ENDIF
\item[] 
\IF{$\texttt{sign} == 0$}
    \STATE $\texttt{sleep} = \texttt{shift} - \texttt{sample}$;
\ELSE
    \STATE $\texttt{sleep} = \texttt{shift} + \texttt{sample}$;
\ENDIF
\item[]
\STATE $\texttt{NOP}(\texttt{sleep})$; \COMMENT{Delay for sampled amount of time}
\STATE $\texttt{HALT}$;
\end{algorithmic}

\caption{Timing-Private Delay Program}
\label{program:disclap-timing-private}
\end{algorithm}

\section{Chaining and Composition}\label{sec:chaining}
Software libraries such as OpenDP and Tumult Core~\cite{gaboardi2020programming,berghel2022tumult,Shoemate_OpenDP_Library} support chaining together multiple algorithms in a modular fashion to create more complex mechanisms with provable stability and privacy guarantees. Motivated by these libraries, we extend our framework to support similar functionality. Specifically, we discuss properties of timing-stable programs that are combined to create more complex functionality such as chaining and composition. 

\subsection{Chaining Timing-Stable Programs}
While timing stability bounds the differences in program execution time for a single program $P$, we are often interested in chaining together multiple programs, e.g., executing program $P_2$ on the output of program $P_1$. For such chaining operations, we use the notation $(P_2\circ P_1)(x, \env_1)$ interchangeably with $P_2(P_1(x, \env_1), \env_2)$ for some execution environment $\env_2$ (see below).



\begin{definition}[Chaining-compatible Programs]\label{def:chaining-compat}
    We say that programs $P_1 : \mathcal{X}\times\mathcal{E}\to\mathcal{Y}\times\mathcal{E}$ and $P_2 : \mathcal{Y}\times\mathcal{E}\to\mathcal{Z}\times\mathcal{E}$ are \emph{chaining-compatible} if there exists a program
    $(P_2\circ P_1)$ and a constant $c$ such that for inputs $x\in \mathcal{X}$ and
    environments $\env_1\in \mathcal{E}$, there exists a (possibly random) $\env_2\in \mathcal{E}$ such that we have:
    $$\out((P_2\circ P_1)(x,\env_1)) = \out(P_2(\out(P_1(x, \env_1)),\env_2))$$
    and
        $$\TPP{(P_2\circ P_1)}{x, \env_1} = \TPP{P_1}{x, \env_1} + \TPP{P_2}{\out(P_1(x, \env_1)), \env_2} + c.$$
    Here we allow the random variable $\env_2$ to be arbitrarily correlated with $x$, as well as the output and runtime of $P_1(x,\env_1)$, but should be independent of the coin tosses of $P_2$.
\end{definition}

\begin{remark}
    If the construction of $P_2\circ P_1$ in the computational model does not modify the environment when chaining programs together, then $\env_2$ equals $\out\env(P_1(x, \env_1))$. However, real-world systems often modify the environment before executing sequential programs, e.g., setting up CPU registers.
\end{remark}

RAM and Word RAM programs can easily be made chaining-compatible. 

\begin{lemma}\label{lemma:chaining-word-ram}
    Word RAM programs $P_1:\mathcal{X} \times\mathcal{E}\to\mathcal{Y}\times\mathcal{E}$ and $P_2: \mathcal{Y}\times\mathcal{E}\to\mathcal{Z}\times\mathcal{E}$ are chaining compatible.
\end{lemma}

\begin{proof}
    Observe that, after $P_1$'s execution, the program's output $y = \out(P_1(x, \env_1))$ is located in memory locations 
    \begin{align*}
        M[\texttt{output\_ptr}],\dots, M[\texttt{output\_len} - 1]
    \end{align*}
    To create the chained program $P_2 \circ P_1$, we replace the $\texttt{HALT}$ instruction at the end of $P_1$ with the instructions $\texttt{input\_ptr = output\_ptr}$ and $\texttt{input\_len = output\_len}$.  We then add all of the code from $P_2$ and adjust the necessary line numbers in any $\texttt{GOTO}$ statements in the relevant $P_2$ code. The resulting program is an execution of $P_2$ on input $y$ in environment $\env_2$ that is the same as $\out\env(P_1(x, \env_1))$ except for these modifications to $\texttt{input\_ptr}$ and $\texttt{input\_len}$. The total runtime is the time to execute $P_1(x, \env_1)$, plus the execution time corresponding to the extra lines for adjusting the input pointer and input length variables, plus the time to execute $P_2$ on the output from $P_1$ in environment $\env_2$. The proof follows similarly for RAM programs.
\end{proof}

\begin{lemma}\label{lemma:chaining-joint}
Suppose $P_1 : \mathcal{X}\times\mathcal{E} \to \mathcal{Y}\times\mathcal{E}$ is $(d_1\mapsto \{d_2, t_1\})$-jointly output/timing stable under input distance metric $d_{\mathcal{X}}$ and output distance metric $d_\mathcal{Y}$. Similarly, suppose $P_2 : \mathcal{Y}\times\mathcal{E}\to \mathcal{Z}\times\mathcal{E}$ is $(d_2 \mapsto \{d_3, t_2\})$-jointly output/timing stable under input distance metric $d_{\mathcal{Y}}$ and output distance metric $d_\mathcal{Z}$. Then if $P_1$ and $P_2$ are chaining compatible, the chained program $P_2\circ P_1: \mathcal{X} \times\mathcal{E} \to \mathcal{Z}\times\mathcal{E}$ 
is $(d_1 \mapsto \{d_3, t_1 + t_2\})$-\emph{jointly output/timing stable}.
\end{lemma}

\begin{proof}
    We show that for every $x, x'$ satisfying $d_{\mathcal{X}}(x, x') \le d_1$ under metric $d_\mathcal{X}$, every pair of input-compatible execution environments $\env_1, \env_1' \in \mathcal{E}$, there exists a coupling $((\Tilde{u}, \Tilde{v}), (\Tilde{u}', \Tilde{v}'))$ of $(\out((P_2 \circ P_1)(x, \env_1)), \TPP{P_2\circ P_1}{x, \env_1})$ and $(\out((P_2 \circ P_1)(x', \env_1')), \TPP{(P_2\circ P_1)}{x', \env'_1})$ such that $d_\mathcal{Z}(\Tilde{u}, \Tilde{u}') \le d_3$ and $|\Tilde{v} - \Tilde{v}'| \le t_1 + t_2$ with probability $1$. \\

    \noindent We construct the coupling as follows. Let $((\Tilde{u}_1, \Tilde{v}_1),(\Tilde{u}'_1, \Tilde{v}'_1))$ be the coupling associated with the joint random variables $(\out(P_1(x, \env_1)), \TPP{P_1}{x, \env_1})$ and $(\out(P_1(x', \env_1')), \TPP{P_1}{x', \env_1'})$ satisfying $d_\mathcal{Y}(\Tilde{u}_1, \Tilde{u}'_1) \le d_2$ and $|\Tilde{v}_1 - \Tilde{v}'_1| \le t_1$ with probability $1$ (by joint output/timing stability of $P_1$).
    
    Then condition on any fixed $((u_1, v_1), (u'_1, v'_1))\in \supp((\Tilde{u}_1, \Tilde{v}_1), (\Tilde{u}'_1, \Tilde{v}'_1))$ and execution environments $\env_2$ and $\env_2'$ (conditioned on $(u_1, v_1)$ and $(u'_1, v'_1)$ respectively). Since $d_\mathcal{Y}(u_1, u'_1) \le d_2$, then for every pair of execution environments $\env_2, \env_2' \in \mathcal{E}$, there exists a coupling $((\Tilde{u}_2, \Tilde{v}_2),(\Tilde{u}'_2, \Tilde{v}'_2))$ of the joint random variables $(\out(P_2(u_1, \env_2)), \TPP{P_2}{u_1, \env_2})$ and $(\out(P_2(u_1', \env_2')), \TPP{P_2}{u_1', \env_2'})$ satisfying $d_\mathcal{Z}(\Tilde{u}_2, \Tilde{u}_2') \le d_3$ and $|\Tilde{v}_2 - \Tilde{v}'_2| \le t_2$ with probability $1$ (by joint output/timing stability of $P_2$). \\

    We sample $((u_2, v_2), (u'_2, v'_2))\sim ((\Tilde{u}_2, \Tilde{v}_2), (\Tilde{u}'_2, \Tilde{v}'_2))|_{\env_2, \env'_2}$ and set:

    \begin{align*}
        (\Tilde{u}, \Tilde{v}) &= (u_2, v_1 + v_2 + c) \\
        &\equiv (\out(P_2(u_1, \env_2)), \TPP{P_1}{x, \env_1} + \TPP{P_2}{u_1, \env_2} + c) \\
        &\equiv (\out(P_2(\out(P_1(x, \env_1)), \env_2)), \TPP{P_1}{x, \env_1} + \TPP{P_2}{\out(P_1(x, \env_1)), \env_2} + c) \\
        &\equiv ((\out((P_2 \circ P_1)(x, \env_1)), \TPP{P_2\circ P_1}{x, \env_1})
    \end{align*}

    and similarly for $(\Tilde{u}', \Tilde{v}')$, where $c$ is the chaining constant in Definition~\ref{def:chaining-compat}. It follows that:
    \begin{align*}
        d_\mathcal{Z}(\Tilde{u}, \Tilde{u}') \le d_3
    \end{align*}
    and
    \begin{align*}
        |\Tilde{v} - \Tilde{v}'| &= |v_1 + v_2 + c - v_1' - v_2' - c| \\
        &= |v_1 - v_1' + v_2 - v_2'| \\
        &\le |v_1 - v_1'| + |v_2 - v_2'| \\
        & \le t_1 + t_2
    \end{align*}
\end{proof}

\begin{lemma}\label{lemma:oc-chaining}
Suppose $P_1 : \mathcal{X}\times\mathcal{E} \to \mathcal{Y}\times\mathcal{E}$ is $(d_1\mapsto \{d_2, t_1\})$-jointly output/timing stable under input distance metric $d_{\mathcal{X}}$ and output distance metric $d_\mathcal{Y}$. Similarly, suppose $P_2 : \mathcal{Y}\times\mathcal{E}\to \mathcal{Z}\times\mathcal{E}$ is $(d_2 \mapsto t_2)$-OC timing stable under $d_{\mathcal{Y}}$. Then if $P_1$ and $P_2$ are chaining compatible, the chained program $(P_2\circ P_1): \mathcal{X} \times\mathcal{E} \to \mathcal{Z}\times\mathcal{E}$ 
is $(d_1 \mapsto t_1 + t_2)$-\emph{OC timing stable}.
\end{lemma}

\begin{proof}
    We show that for every $x, x'$ satisfying $d_{\mathcal{X}}(x, x') \le d_1$ under metric $d_\mathcal{X}$, every pair of input-compatible execution environments $\env_1, \env_1' \in \mathcal{E}$, and for every $z \in \supp((P_2\circ P_1)(x, \env_1)) \cap \supp((P_2\circ P_1)(x', \env_1'))$, there exists a coupling $(\Tilde{w}, \Tilde{w}')$ of $\TPP{P_2\circ P_1}{x, \env_1}_{|\out((P_2\circ P_1)(x, \env_1)) = z}$ and $\TPP{(P_2\circ P_1)}{x', \env'_1}_{|\out((P_2\circ P_1)(x', \env_1')) = z}$ such that $|\Tilde{w} - \Tilde{w}'| \le t_1 + t_2$ with probability $1$. \\

    \noindent We construct the coupling as follows. Let $((\Tilde{u}, \Tilde{v}),(\Tilde{u}', \Tilde{v}'))$ be the coupling associated with the joint random variables $(\out(P_1(x, \env_1)), \TPP{P_1}{x, \env_1})$ and $(\out(P_1(x', \env_1')), \TPP{P_1}{x', \env_1'})$ satisfying $d_\mathcal{Y}(\Tilde{u}, \Tilde{u}') \le d_2$ and $|\Tilde{v} - \Tilde{v}'| \le t_1$ with probability $1$ (by joint output/timing stability). \\
    
    \noindent Then condition on any fixed $((u, v), (u', v'))\in \supp((\Tilde{u}, \Tilde{v}),(\Tilde{u}', \Tilde{v}'))$ and execution environments $\env_2$ and $\env_2'$ (conditioned on $(u, v)$ and $(u', v')$ respectively). Since $d_\mathcal{Y}(u, u') \le d_2$, then for every $z \in \supp(\out(P_2(u, \env_2)))\cap \supp(\out(P_2(u', \env_2')))$, there exists a coupling $(\Tilde{r}, \Tilde{r}')$ of the conditional runtime random variables $\TPP{P_2}{u, \env_2}|_{\out(P_2(u, \env_2)) = z}$ and $\TPP{P_2}{u', \env_2'}|_{\out(P_2(u', \env_2')) = z}$ satisfying $|\Tilde{r} - \Tilde{r}'| \le t_2$ (by output-conditional timing stability). \\
    
    \noindent We sample $(r, r') \sim (\Tilde{r}, \Tilde{r}')|_{\env_2, \env'_2}$ and set:
    \begin{align*}
        \Tilde{w} &= v + r + c\\
        &\equiv \TPP{P_1}{x, \env_1} + \TPP{P_2}{u, \env_2}|_{\out(P_2(u, \env_2)) = z} + c\\
        &\equiv \TPP{P_1}{x, \env_1} + \TPP{P_2}{\out(P_1(x, \env_1)), \env_2}|_{\out(P_2\circ P_1(x, \env_1)) = z} + c\\
        &\equiv \TPP{P_2\circ P_1}{x}|_{P_2(P_1(x)) = z}
    \end{align*}

    and similarly for $w'$, where $c$ is the chaining constant in Definition~\ref{def:chaining-compat}. It follows that 
    \begin{align*}
        |\Tilde{w} - \Tilde{w}'| &= |u + r + c - u' - r' - c| \\
        &= |u - u' + r - r'| \\
        &\le |u - u'| + |r - r'| \\
        &\le t_1 + t_2
    \end{align*} 
\end{proof}

Finally, we can chain together OC-timing-stable programs with timing-private delay programs to achieve timing privacy:

\begin{theorem}\label{thm:main-theorem}
    Let $P_1 : \mathcal{X}\times\mathcal{E}\to\mathcal{Y}\times\mathcal{E}$ be $(d_1\mapsto t_1)$-OC timing stable under input distance metric $d_\mathcal{X}$. If $P_2 : \mathcal{Y}\times\mathcal{E} \to \mathcal{Y}\times\mathcal{E}$ is a $(t_1 \mapsto d_2)$-timing-private delay program with respect to privacy measure $\mathcal{M}$, then $P_2\circ P_1$ is $(d_1 \mapsto d_2)$-timing private with respect to $d_\mathcal{X}$ and $\mathcal{M}$.
\end{theorem}

\begin{proof}
     We want to show that for all pairs of $x, x'$ such that $d_{\mathcal{X}}(x, x') \le d_1$, all pairs of input-compatible environments $\env_1, \env_1'$, and all $y \in \supp(\out(P_1(x, \env_1))) \cap \supp(\out(P_1(x', \env_1')))$, 
    \begin{align*}
        \mathcal{M}(\TPP{(P_2 \circ P_1)}{x, \env_1}_{|\out((P_2\circ P_1)(x,\env_1)) = y}, \TPP{(P_2\circ P_1)}{x', \env_1'}|_{\out((P_2\circ P_1)(x',\env_1')) = y}) \le d_2
    \end{align*}
     
     Let $\Phi$ be the $(t_1 \mapsto d_2)$-timing private delay distribution implemented by $P_2$, $c$ is the chaining constant used in Definition~\ref{def:chaining-compat}, and $(\Tilde{r}, \Tilde{r}')$ is the coupling of the conditional random variables $\TPP{P_1}{x, \env_1}_{|\out(P_1(x, \env_1)) = y}$ and $\TPP{P_1}{x', \env_1'}_{|\out(P_1(x', \env_1')) = y}$ satisfying $|\Tilde{r} - \Tilde{r}'|\le t_1$ (by output-conditional timing stability). Let $(r, r') \sim (\Tilde{r}, \Tilde{r}')$ and $\env_2$ and $\env_2$ be conditioned on $(r, y)$ and $(r', y)$ respectively. Then,
    \begin{align*}
        r + c + \Phi &\equiv \TPP{P_1}{x, \env_1}_{|\out(P_1(x, \env_1)) = y} + c + \phi \\
        &\equiv \TPP{P_1}{x, \env_1}_{|\out(P_1(x, \env_1)) = y} + c + \TPP{P_2}{y, \env_2}_{|\out(P_1(x, \env_1)) = y} \\
        &\equiv \TPP{(P_2 \circ P_1)}{x, \env_1}_{|\out((P_2\circ P_1)(x,\env_1)) = y}
    \end{align*}
    and similarly $r' + c + \Phi \equiv \TPP{(P_2\circ P_1)}{x', \env_1'}_{|\out((P_2\circ P_1)(x',\env_1')) = y}$. Since $\Pr[|r - r'| \le t_1] = 1$, and $\Phi$ is a timing-private delay distribution, it follows that for $T\sim \Phi$:

    \begin{align*}
        \mathcal{M}(r + c + T, r' + c + T) \le d_2 \\
    \end{align*}
    and therefore
    \begin{align*}
        \mathcal{M}(\TPP{(P_2 \circ P_1)}{x, \env_1}_{|\out((P_2\circ P_1)(x,\env_1)) = y}, \TPP{(P_2\circ P_1)}{x', \env_1'}_{|\out((P_2\circ P_1)(x',\env_1')) = y}) \le d_2
    \end{align*}
\end{proof}

We now give prove a post-processing property for programs.

\begin{lemma}[Post-processing Program Outputs]\label{lemma:postprocess-program-outputs}
Let $P_1: \mathcal{X}\times\mathcal{E}\to \mathcal{Y}\times\mathcal{E}$ be $(\varepsilon, \delta)$-differentially private program under input metric $d_\mathcal{X}$ and the smoothed max-divergence privacy measure. Let $P_2: \mathcal{Y} \times\mathcal{E}\to \mathcal{Z}\times\mathcal{E}$ be a (possibly) randomized program that is output-pure (Definition~\ref{def:output-pure-programs}).
Then if $P_2$ and $P_1$ are chaining compatible, the chained program $P_2\circ P_1: \mathcal{X} \times\mathcal{E} \to \mathcal{Z}\times\mathcal{E}$ is $(\varepsilon, \delta)$-differentially private under input metric $d_\mathcal{X}$ and the smoothed max-divergence privacy measure $D_\infty^\delta$.
\end{lemma}

\begin{proof}
    For all adjacent inputs $x,x'\in\mathcal{X}$ under input metric $d_\mathcal{X}$, and all $\env_1, \env_1' \in \mathcal{E}$, we have that $D_\infty^\delta(\out(P_1(x, \env_1))|| \out(P_1(x', \env'_1)) \le \varepsilon$ and $D_\infty^\delta(\out(P_1(x', \env_1'))||\out(P_1(x, \env_1)) \le \varepsilon$ by the fact that $P_1$ is achieves $(\varepsilon, \delta)$-differential privacy (Lemma~\ref{lemma:smoothed-divergence-dp}). Since $P_2$ is output-pure, for all execution environments $\env_2 \in \mathcal{E}$, $\out(P_2(\out(P_1(x, \env_1)), \env_2) = f(\out(P_1(x, \env_1)))$ where $f:\mathcal{Y}\to\mathcal{Z}$ is a (possibly randomized) function. By the post-processing property for approximate DP (Lemma~\ref{lem:dp postprocess}), it follows that for all $\env_2, \env_2' \in \mathcal{E}$:
    \begin{align*}
        D_\infty^\delta(\out(P_2(\out(P_1(x, \env_1)), \env_2))||\out(P_2(\out(P_1(x', \env'_1)), \env_2'))) \\
        &\hspace{-10cm}= D_\infty^\delta(f(\out(P_1(x, \env_1)))|| f(\out(P_1(x', \env'_1)))) \\
        &\hspace{-10cm}\le D_\infty^\delta(\out(P_1(x, \env_1))|| \out(P_1(x', \env'_1))) \\
        &\hspace{-10cm}\le \varepsilon
    \end{align*}

    and similarly for $D_\infty^\delta(\out(P_2(\out(P_1(x', \env_1')), \env_2'))|| \out(P_2(\out(P_1(x, \env_1)), \env_2)))$.
\end{proof}

\begin{lemma}\label{lemma:timing-private-delay-post-process}
    Let $P_1: \mathcal{X}\times\mathcal{E}\to\mathcal{Y}\times\mathcal{E}$ be $(\din \mapsto \dout)$-differentially private with respect to input distance metric $d_\mathcal{X}$ and a privacy measure $M$. Let $P_2:\mathcal{Y}\times\mathcal{E}\to\mathcal{Y}\times\mathcal{E}$ be a timing-private delay program such that $P_1$ and $P_2$ are chaining compatible. Then the chained program $P_2 \circ P_1$ is $(\din\mapsto \dout)$-differentially private with respect to input metric $d_\mathcal{X}$ and privacy measure $M$.
\end{lemma}

\begin{proof}
    The proof follows from the fact that timing-private delay programs are output pure (Lemma~\ref{lemma:timing-private-delay-programs-are-output-pure}) and Lemma~\ref{lemma:postprocess-program-outputs}.
\end{proof}

The post-processing property for DP (Lemma ~\ref{lem:dp postprocess}) states that arbitrary transformations applied to differentially private outputs produces new outputs that remain differentially private. Unfortunately, our definition of timing privacy is incompatible with post-processing. Timing privacy bounds the \emph{additional} information leakage caused by observing the program's running time after observing the program's output. However, post-processing can potentially destroy information in the program's output and this can result in unbounded additional information leakage in the program's runtime.

To see this effect in action, consider a jointly output/timing-pure program $P$ implementing the identity function $f(x) = x$ with the additional property that the runtime reveals something about the input. That is, $\TP{x} = g(x)$ for a non-constant function $g$. Then $P$ is $(1\mapsto 0)$-timing-private with respect to the change-one distance metric $d_{CO}$ and any privacy measure $\mathcal{M}$ (e.g., max divergence). The identity function program satisfies perfect timing privacy despite its running time being directly correlated with its input. This is due to the fact that $P$'s output already reveals the input and so the program's running time reveals no additional information. However, if we apply the simple post-processing step by chaining $P$ with a deterministic, constant-time program $P'$ that sets $\texttt{output\_len} = 0$, then the result is no longer perfectly timing-private since $\out((P' \circ P)(x, \env)) = \lambda\footnote{We use $\lambda$ to indicate an empty output}$ while $\TPP{P'\circ P}{x, \env}$ leaks information about the input (everything about it if $g$ is injective). 

%

\subsection{A Timing-Private Unbounded Noisy Sum}

We illustrate the chaining operations above by chaining the Sum program (Program~\ref{program:sum}) with the Discrete Laplace Mechanism (Program~\ref{program:discrete-laplace}, given below) and our Timing-Private Delay Program (Program~\ref{program:disclap-timing-private}).  

We start by giving a Word RAM program for the Discrete Laplace Mechanism and analyze its timing properties. 

\begin{algorithm}[h!]
\vspace{5px}
\textbf{Input:} A number $y \in \mathbb{N}$ occupying memory location $M[\texttt{input\_ptr}]$

\vspace{5px}
\begin{flushleft}
    \textbf{Output:} $\max \{y + \textrm{DiscreteLaplace}(\mu = 0, s), 0\}$ where DiscreteLaplace is defined as in Section~\ref{sub-sec:delays} and
    $s=1/\ln(\texttt{b}/\texttt{a})$. For the Word RAM version of the program, we require $s < 2^\omega$ and $\texttt{a} + \texttt{b} < 2^\omega$ and $\min\{\max \{y + \textrm{DiscreteLaplace}(\mu = 0, s), 0\}, 2^\omega - 1\}$ is output instead since the Word RAM program outputs values in $[0, 2^\omega)$.
\end{flushleft}

\vspace{5px}

\begin{algorithmic}[1]
\STATE $\texttt{output\_len} = 1$;
\STATE $\texttt{output\_ptr} = 0$;
\STATE $\texttt{y} = M[\texttt{input\_ptr}]$;
\STATE $\texttt{noise} = 0$;
\STATE $\texttt{set} = 0$;
\STATE $\texttt{sign} = \texttt{RAND}(1)$; \COMMENT{Pick a uniformly random noise direction (positive or negative)}
\STATE \texttt{zprobA} = $\texttt{b} - \texttt{a}$;
\STATE \texttt{zprobB} = $\texttt{b} + \texttt{a}$;
\STATE $\texttt{idx} = \texttt{RAND}(\texttt{zprobB} - 1)$;
\IF{$\texttt{idx} < \texttt{zprobA}$}
    \STATE $\texttt{set} = 1$; \COMMENT{sample $0$ noise with probability $\frac{\texttt{b} - \texttt{a}}{\texttt{b} + \texttt{a}}$}
\ELSE
    \STATE $\texttt{set} = 0$;
\ENDIF
\item[]
\WHILE{$\texttt{set} == 0$}
\STATE $\texttt{noise} = \texttt{noise} + 1$; \COMMENT{sample from a Geometric random variable with $p = \frac{\texttt{b} - \texttt{a}}{\texttt{b}}$}
\STATE $\texttt{idx} = \texttt{RAND}(\texttt{b} - 1)$;
\IF{$\texttt{idx} < \texttt{b} - \texttt{a}$}
\STATE $\texttt{set} = 1$;
\ENDIF
\ENDWHILE
\item[]
\IF{$\texttt{sign} == 0$}
\STATE $\texttt{noisy\_y} = \texttt{y} - \texttt{noise}$;
\ELSE
\STATE $\texttt{noisy\_y} = \texttt{y} + \texttt{noise}$;
\ENDIF
\item[]
\STATE $M[\texttt{output\_ptr}] = \texttt{noisy\_y}$;
\STATE $\texttt{HALT}$;
\end{algorithmic}

\caption{Discrete Laplace Mechanism}\label{program:discrete-laplace}
\end{algorithm}

\begin{lemma}\label{lemma:disc-laplace-measurement-output-pure}
    The Discrete Laplace Word RAM program $P: \mathbb{N}\times\mathcal{E}\to\mathbb{N}\times\mathcal{E}$  (Program~\ref{program:discrete-laplace}) is output-pure and timing-pure.
\end{lemma}
\begin{proof}
    The program only ever reads memory at location $M[\texttt{input\_ptr}]$ (line 3) to obtain the input and later writes the output to $M[0]$. Since the program's execution does not depend on any memory values beyond its input, then for all $y \in \mathbb{N}$,  the program's output and runtime is identically distributed on all pairs of input-compatible execution environments $\env, \env' \in\mathcal{E}$.
\end{proof}

\begin{lemma}\label{lemma:laplace-output-privacy}
    Let $P:\mathbb{N}\times\mathcal{E}\to\mathbb{N}\times\mathcal{E}$ be the Discrete Laplace program (Program~\ref{program:discrete-laplace}) with scale parameter $s = 1/\ln(\texttt{b}/\texttt{a}) = \din / \varepsilon$. Then $P$ is $(\din\mapsto \varepsilon)$-DP with respect to input metric $d_\mathbb{N}(y, y') = |y - y'|$ and the max-divergence privacy measure.
\end{lemma}

\begin{proof}
    We will show that for all for $y, y'$ satisfying $|y - y'| \le \dout$, and all $S \subseteq\mathbb{N}$, $\Pr[\out(P_2(y)) \in S]/\Pr[\out(P_2(y')) \in S] \le e^\varepsilon$. We ignore the execution environments in our analysis since $P$ is output-pure (Lemma~\ref{lemma:disc-laplace-measurement-output-pure}).

    The Discrete Laplace program first samples a value from a censored Discrete Laplace distribution parameterized by a shift $\mu = y$ and scale $s = 1/\ln(\texttt{b}/\texttt{a}) = \dout/\varepsilon$. The program works similarly to Program~\ref{program:disclap-timing-private} in how it samples from a Discrete Laplace distribution except that it does not bound the number of trials when sampling from a Geometric distribution (lines 14-18). Since the Word RAM model does not support values $< 0$ or $> 2^\omega - 1$, the result of $\texttt{noisy\_y} = \texttt{y} - \texttt{noise}$ (line 20) is always rounded to $0$ when $\texttt{noise} > \texttt{y}$ and the result of $\texttt{noisy\_y} = \texttt{y} + \texttt{noise}$ (line 22) is always rounded to $2^\omega - 1$. As a result, the program's output is distributed as a \emph{censored} Discrete Laplace random variable with PMF $f$ defined below:
    \begin{align*}
        \text{$f(z |\mu = y, s = 1/\ln(\texttt{b}/\texttt{a}))$} =
        \begin{cases}
        \frac{\texttt{b} -\texttt{a}}{\texttt{b} + \texttt{a}} & z = \mu \\
        \frac{1}{2}\cdot \big(1 -\frac{\texttt{b} -\texttt{a}}{\texttt{b} + \texttt{a}}\big)\cdot \big(\frac{\texttt{a}}{\texttt{b}}\big)^{|z - \mu| - 1}\cdot \big(1 - \frac{\texttt{a}}{\texttt{b}}\big) & 0 < z < 2^\omega, z\ne \mu \\
        \frac{1}{2}\cdot \big(1 -\frac{\texttt{b} -\texttt{a}}{\texttt{b} + \texttt{a}}\big)\cdot \big(\frac{\texttt{a}}{\texttt{b}}\big)^{|z - \mu| - 1} & z = 0 \textrm{ or } z = 2^\omega - 1  \\
        0 & \textrm{otherwise}
        \end{cases}
    \end{align*}

    It follows that for all $y \in \mathcal{Y}$:
    \begin{align*}
        \Pr[\out(P_2(y)) = y] &= \frac{\texttt{b} - \texttt{a}}{\texttt{b} + \texttt{a}} \\
        &= \frac{e^{\ln(\texttt{b}/\texttt{a})} - 1}{e^{\ln(\texttt{b}/\texttt{a})} + 1}\\
        &= \frac{e^{1/s} - 1}{e^{1/s} + 1} \\
        &= \frac{e^{1/s} - 1}{e^{1/s} + 1}\cdot e^{-|y - y|/ s}
    \end{align*}
   
   which matches the probability mass function for a Discrete Laplace distribution. Similarly, for all $y, z \in \mathbb{N}$, $0 < z < 2^\omega - 1, z\ne y$:
    \begin{align*}
        \Pr[\out(P_2(y)) = z] &= \frac{1}{2}\cdot \big(1 - \frac{\texttt{b} - \texttt{a}}{\texttt{b} + \texttt{a}}\big)\cdot \big(\frac{\texttt{a}}{\texttt{b}}\big)^{|z - y| - 1}\cdot (1 - \frac{\texttt{a}}{\texttt{b}}) \\
        &= \frac{1}{2}\cdot(1 - \frac{e^{\ln(\texttt{b}/\texttt{a})} - 1}{e^{\ln(\texttt{b}/\texttt{a})} + 1})\cdot (e^{-\ln(\texttt{b}/\texttt{a})})^{(|z - y| - 1)} \cdot (1 - e^{-\ln(\texttt{b}/\texttt{a})}) \\
        &= \frac{1}{2}\cdot(1 - \frac{e^{1/s} - 1}{e^{1/s} + 1})\cdot (e^{-1/s})^{(|z - y| - 1)} \cdot (1 - e^{-1/s}) \\
        &= \frac{1}{2}\cdot (\frac{e^{1/s} + 1 - e^{1/s} + 1}{e^{1/s} + 1})\cdot (e^{-1/s})^{(|z - y| - 1)} \cdot (1 - e^{-1/s}) \\
        &= \frac{1}{e^{1/s} + 1}\cdot (e^{-1/s})^{(|z - y| - 1)} \cdot (1 - e^{-1/s}) \\
        &= \frac{e^{(-|z - y| + 1)/s}- e^{-(|z - y|)/s}}{e^{1/s} + 1} \\
        &= \frac{e^{1/s} - 1}{e^{1/s} + 1}\cdot e^{-|z - y|/s}
    \end{align*}

    Therefore, for all $y, y' \in \mathbb{N}$ such that $|y - y'| \le \dout$, and all $0 < z < 2^\omega - 1$, we have:

    \begin{align*}
        \frac{\Pr[\out(P_2(y)) = z]}{\Pr[\out(P_2(y')) = z]} &= \frac{e^{-|z - y|/s}}{e^{-|z - y'|/s}} \\
        &= e^{(|z - y'| - |z - y|)/s} \\
        &\le e^{|y - y'|/s} \\
        &\le e^{\varepsilon \cdot |y - y'|/\dout} \\
        &\le e^\varepsilon 
    \end{align*}

    Finally for $z = 2^\omega - 1$ or $z = 0$, 
    \begin{align*}
        \frac{\Pr[\out(P_2(y)) = z]}{\Pr[\out(P_2(y')) = z]} &= \frac{(\texttt{a}/\texttt{b})^{|z - y| - 1}}{(\texttt{a}/\texttt{b})^{|z - y'| - 1}} \\
        &= \frac{e^{-\ln(\texttt{b}/\texttt{a})\cdot (|z - y| - 1)}}{e^{-\ln(\texttt{b}/\texttt{a})\cdot (|z - y'| - 1)}} \\
        &= \frac{e^{-(|z - y| - 1)/s}}{e^{-(|z - y'| - 1)/s}} \\
        &\le e^{(|z - y'| - |z - y|)/s} \\
        &\le e^{\varepsilon\cdot |y - y'|/\dout} \\
        &\le e^\varepsilon
    \end{align*}
    
     Therefore, for all $y, y'$ satisfying $|y - y'| \le \dout$, and all $S \subseteq \mathbb{N}$, we have 
     \begin{align*}
         \frac{\Pr[\out(P_2(y)) \in S]}{\Pr[\out(P_2(y')) \in S]} \le e^\varepsilon
     \end{align*}
     
     and the claim is satisfied.
\end{proof}

\begin{lemma}\label{lemma:noisy-sum-output-private}
    Let $P_1 : \mathcal{X}\times\mathcal{E}\to\mathbb{N}\times\mathcal{E}$ be the Sum  Word RAM program (Program~\ref{program:sum}) and $P_2:\mathbb{N}\times\mathcal{E}\to\mathbb{N}\times\mathcal{E}$ be the Discrete Laplace (Program~\ref{program:discrete-laplace}) with scale parameter $s = 1/\ln(\texttt{b}/\texttt{a})$ satisfying $\ln(\texttt{b}/\texttt{a}) = \varepsilon/\Delta$. Then the chained program $P_2 \circ P_1$ is $\varepsilon$-DP under the insert-delete distance metric $\did$.
\end{lemma}
\begin{proof}
    The proof follows from the fact that $P_1$ is $(1\mapsto \Delta)$-output stable (Lemma~\ref{lemma:joint-stability-sum} and Lemma~\ref{lemma:joint-implies-output-stable}) and therefore $P_2 \circ P_1$ is the Discrete Laplace mechanism executed on $\Delta$-close inputs. The proof follows from Lemma~\ref{lemma:laplace-output-privacy}.
\end{proof}
\begin{lemma}\label{lemma:oc-stability-disclap}
    For any $\din \in \mathbb{N}$, the Discrete Laplace program $P: \mathbb{N}\times\mathcal{E}\to\mathbb{N}\times\mathcal{E}$  (Program~\ref{program:discrete-laplace}) is $(\din \mapsto 5\cdot \din)$-OC timing stable under the input distance metric $d_\mathbb{N}$ defined as $d_\mathbb{N}(x, x') = |x - x'|$.
\end{lemma}

\begin{proof}
    Conditioned on $\out(P(x, \env)) = y$, the program has deterministic runtime. The program's conditional runtime $\TP{x, \env}|_{\out(P(x,\env)) = y}$ is equal to $15 + (5\cdot |x - y|)$ where $5\cdot |x - y|$ comes from the number of loops that are executed in lines (14-18). The other $15$ instructions are always executed on every input (lines 1-10, 2 instructions depending on the branching condition in lines 10-13, 2 instructions depending on the branching condition in lines 19-22, and lines 23-24). Therefore, for all inputs $x, x' \in \mathbb{N}$, conditioned on output $y \in \mathbb{N}$,
    \begin{align*}
        |\TP{x, \env}|_{\out(P(x,\env)) = y} - \TP{x, \env}|_{\out(P(x,\env)) = y}| &= |15 + 5\cdot|x - y| - 15 - 5\cdot |x' - y|| \\
        &= |5\cdot (|x - y| - |x' - y|)|\\
        &\le 5\cdot (|x  - x'|)\\
        &= 5\cdot \din
    \end{align*}
\end{proof}

\begin{lemma}\label{lemma:output-stability-sum}
    Let $P_1 :\mathcal{X}\times\mathcal{E}\to\mathbb{N}\times\mathcal{E}$ be the Sum program (Program~\ref{program:sum}) and $P_2:\mathbb{N}\times\mathcal{E}\to\mathbb{N}\times\mathcal{E}$ be the Discrete Laplace program (Program~\ref{program:discrete-laplace}). Then $P_2 \circ P_1$ is $(1\mapsto 3 + 5\Delta)$-OC timing stable under the insert-delete distance metric $\did$ on the input and the output distance metric $d_\mathbb{N}$ defined as $d_\mathbb{N}(y, y') = |y - y'|$.
\end{lemma}

\begin{proof}
    Since $P_1$ is $(1\mapsto \{\Delta, 3\})$-jointly output/timing stable (Lemma~\ref{lemma:joint-stability-sum}), $P_2$ is $(\din\mapsto 5\cdot\din)$-OC timing stable (Lemma~\ref{lemma:oc-stability-disclap}), and $P_1$ is chaining compatible with $P_2$ (by Lemma~\ref{lemma:chaining-word-ram}), then by Lemma~\ref{lemma:oc-chaining}, the chained program $P_2 \circ P_1$ is $(1\mapsto 3 + 5\cdot\Delta)$-OC-timing stable. 
\end{proof}


\begin{lemma}\label{lemma:timing-private-sum}
    Let $P_1:\mathcal{X}\times\mathcal{E}\to\mathbb{N}\times\mathcal{E}$ be Program~\ref{program:sum} (Sum), $P_2:\mathbb{N}\times\mathcal{E}\to\mathbb{N}\times\mathcal{E}$ be Program~\ref{program:discrete-laplace} (Discrete Laplace) with scale parameter $s_1 = \Delta/\varepsilon_1$, and $P_3:\mathcal{Z}\times\mathcal{E}\to\mathcal{Z}\times\mathcal{E}$ be Program~\ref{program:disclap-timing-private} (Timing-Private Delay) with scale parameter $s_2 = (3 + 5\Delta)/\varepsilon_2$ and shift $\mu$. Then $P_3 \circ (P_2 \circ P_1)$ is $\varepsilon_1$-differentially private in its output and $(\varepsilon_2, \delta)$-timing private for $\delta = 2\cdot e^{-\varepsilon_2(\mu - (3+5\Delta))/(3+5\Delta)}$.
\end{lemma}

\begin{proof}
    The chained program $P_2\circ P_1 : \mathcal{X}\times\mathcal{E}\to\mathbb{N}\times\mathcal{E}$ is $\varepsilon_1$-differentially private by Lemma~\ref{lemma:noisy-sum-output-private}. Since $P_3$ is a timing-private delay program (Lemma~\ref{lemma:disclap-timing-private-delay}) it will compute a post-processing (identity function) on the output of $P_2 \circ P_1$ and by Lemma~\ref{lemma:timing-private-delay-post-process} the  chained program $P_3 \circ (P_2 \circ P_1):\mathcal{X}\times\mathcal{E}\to\mathbb{N}\times\mathcal{E}$ is $\varepsilon_1$-differentially private.
    
    Since $P_2 \circ P_1$ is $(1\mapsto 3+5\Delta)$-OC timing stable (Lemma~\ref{lemma:output-stability-sum}), and $P_3$ is a $(3+5\Delta\mapsto (\varepsilon_2, \delta))$-timing-private delay program for $\delta = 2\cdot e^{-\varepsilon_2(\mu - (3+5\Delta))/(3+5\Delta)}$ (Lemma~\ref{lemma:disclap-timing-private-delay}) under input distance metric $d_\mathbb{N}$ and the smoothed max-divergence privacy measure $D_{\infty}^\delta$, then by Theorem~\ref{thm:main-theorem}, the chained program $P_3 \circ(P_2 \circ P_1)$ is $(1\mapsto (\varepsilon_2, \delta))$-timing-private with respect to privacy measure $D_\infty^\delta$.
\end{proof}

Note that Lemma~\ref{lemma:timing-private-sum} provides an $\varepsilon_1$-DP mechanism that achieves $\varepsilon_2$-timing privacy in the upper-bounded DP model when the programs are Word RAM programs. Moreover, this construction is much more efficient than the naive approach of padding execution time to the worst-case dataset size $n_{\textrm{max}} = 2^\omega - 1$, which would result in slow runtimes for datasets with size $n < n_{\textrm{max}}$. In the RAM model of computation, the construction yields a timing-private DP sum in the unbounded DP model.

\subsection{Composition of Timing-Private Programs}
We can treat the composition of timing-private mechanisms similarly to how we treat chaining. In particular, our framework supports the composition of timing-private programs that perform intermediate DP computations.

\begin{definition}[Composition-compatible Programs]\label{def:composition-compat}
    We say that programs $P_1 : \mathcal{X}\times\mathcal{E}\to(\mathcal{X}\times\mathcal{Y})\times\mathcal{E}$ and $P_2 : (\mathcal{X}\times\mathcal{Y})\times\mathcal{E}\to\mathcal{Z}\times\mathcal{E}$ are \emph{composition-compatible} if 
    there is a program $P_2\otimes P_1:\mathcal{X}\times\mathcal{E}\to(\mathcal{Y}\times\mathcal{Z})\times\mathcal{E}$ and a constant $c$ such that for all inputs $x\in \mathcal{X}$ and all execution environments
    $\env_1 \in \mathcal{E}$, there exists a (possibly random) input-compatible execution environment $\env_2\in \mathcal{E}$ 
    such that we have:
    \begin{align*}
        \out((P_2\otimes P_1)(x,\env_1))= (\out(P_1(x, \env_1)),\out(P_2((x, \out(P_1(x, \env_1))),\env_2)))
    \end{align*}
    and 
    \begin{align*}
        \TPP{(P_2\otimes P_1)}{x, \env_1}=\TPP{P_1}{x, \env_1} + \TPP{P_2}{(x, \out(P_1(x, \env_1))), \env_2} + c.
    \end{align*}

Similar to the case of chaining-compatible programs, we allow the random variable $\env_2$ to be arbitrarily correlated with the output and runtime of $P_1(x, \env_1)$, but it should be independent of the coin tosses of $P_2$.
\end{definition}

\begin{lemma}[Composition-compatible Word RAM Programs]\label{lemma:comp-compat-word-ram}
    Let program $P_1:\mathcal{X}\times\mathcal{E}\to(\mathcal{X}\times\mathcal{Y})\times\mathcal{E}$ be a RAM program that writes its unmodified input to 
    \begin{align*}
        M[\texttt{output\_ptr}],\dots,M[\texttt{output\_ptr} + \texttt{input\_len} - 1]
    \end{align*}
    and appends the rest of its output to  $M[\texttt{output\_ptr} + \texttt{input\_len}]$, $\dots$, $M[\texttt{output\_ptr} +\texttt{output\_len} - 1]$.
    Then $P_1$ and $P_2$ are \emph{composition-compatible} for all programs $P_2: (\mathcal{X}\times\mathcal{Y})\times\mathcal{E}\to\mathcal{Z}\times\mathcal{E}$.
\end{lemma}

\begin{proof}
    The proof for composition-compatible programs works similarly to that of chaining-compatible programs. We replace the $\texttt{HALT}$ at the end of $P_1$ to set $\texttt{input\_len}=\texttt{output\_len}$ and $\texttt{input\_ptr} = \texttt{output\_ptr}$.  The resulting program is an execution of $P_2$ on input $(x, y)$ in environment $\env_2$ that is the same as $\out\env(P_1(x, \env_1))$ except for these modifications to $\texttt{input\_ptr}$ and $\texttt{input\_len}$.
\end{proof}

\begin{remark}
    The proof follows for RAM programs also.
\end{remark}

\begin{lemma}[Composition of OC-Timing-Stable Programs]\label{lemma:composition-oc-stable}
    Let $P_1 :\mathcal{X}\times\mathcal{E}\to(\mathcal{X}\times\mathcal{Y})\times\mathcal{E}$ be $(d_1 \mapsto t_1)$-OC timing stable with respect to the $\mathcal{Y}$ output coordinate\footnote{This means that we condition only on the $y\in \mathcal{Y}$ part of the output when considering OC-timing stability.} and input distance metric $d_\mathcal{X}$. Similarly, let $P_2:(\mathcal{X}\times\mathcal{Y})\times\mathcal{E} \to(\mathcal{Y}\times\mathcal{Z})\times\mathcal{E}$ be $(d_1 \mapsto t_2)$-OC timing stable with respect to input distance metric $d_\mathcal{X}$ on its first input coordinate. If $P_1$ and $P_2$ are composition compatible, then the composed program $P_2 \otimes P_1$ is $(d_1 \mapsto t_1 + t_2)$-OC timing stable.
\end{lemma}

\begin{proof}
    For all $x, x' \in \mathcal{X}$ satisfying $d_\mathcal{X}(x, x') \le d_1$, all input-compatible environments $\env_1, \env_1' \in \mathcal{E}$, and all $(y, z) \in \supp(\out((P_2\otimes P_1)(x, \env_1)))\cap\supp(\out((P_2\otimes P_1)(x', \env_1')))$, we will construct a coupling $(\Tilde{r}, \Tilde{r}')$ of $\TPP{(P_2 \otimes P_1)}{x, \env_1}|_{\out((P_2 \otimes P_1)(x, \env_1)) = (y, z)}$ and $\TPP{(P_2 \otimes P_1)}{x', \env_1'}|_{\out((P_2 \otimes P_1)(x', \env_1')) = (y, z)}$ such that $|\Tilde{r} - \Tilde{r}'| \le t_1 + t_2$ with probability $1$. \\

    Observe that for all $x, x'\in\mathcal{X}$ such that $d_\mathcal{X}(x, x')\le d_1$, all $\env_1, \env_1'\in\mathcal{E}$, and all $(\cdot, y) \in \supp(\out(P_1(x, \env_1))) \cap \supp(\out(P_1(x', \env_1')))$, there exists a coupling $(\Tilde{r}_1, \Tilde{r}_1')$ of the conditional random variables $\TPP{P_1}{x, \env_1}|_{\out(P_1(x, \env_1)) = (x,y)}$ and $\TPP{P_1}{x', \env_1'}|_{\out(P_1(x', \env_1')) = (x', y)}$ such that $|\Tilde{r}_1 - \Tilde{r}_1'| \le t_1$ (by the OC timing stability of $P_1$). \\

    \noindent For all such $x, x'$ and $\env_1, \env_1'$, take $\env_2$ and $\env_2'$ to be the execution environments satisfying:
    \begin{align*}
        &\TPP{(P_2 \otimes P_1)}{x, \env_1} = \TPP{P_1}{x, \env_1} + \TPP{P_2}{(x, \out(P_1(x, \env_1))), \env_2} + c 
    \end{align*}
    and
    \begin{align*}
        &\TPP{(P_2 \otimes P_1)}{x', \env_1'} = \TPP{P_1}{x', \env_1'} + \TPP{P_2}{(x', \out(P_1(x', \env_1'))), \env_2'} + c 
    \end{align*}
    respectively (such an $\env_2$ and $\env_2'$ exist due to $P_1$ and $P_2$ being composition-compatible). Now fix $y \in \supp(\out(P_1(x, \env_1))) \cap \supp(\out(P_1(x', \env_1')))$. Then for all $z$ such that 
    \begin{align*}
        (y, z) \in \supp(\out(P_2((x, y), \env_2))) \cap \supp(\out(P_2((x', y), \env_2')))
    \end{align*} there exists a coupling $(\Tilde{r}_2, \Tilde{r}_2')$ of the random variables $\TPP{P_2}{(x, y), \env_2}|_{\out(P_2((x, y), \env_2)) = (y, z)}$ and $\TPP{P_2}{(x', y), \env_2'}|_{\out(P_2((x', y), \env_2')) = (y, z)}$ such that $|\Tilde{r}_2 - \Tilde{r}_2'| \le t_2$ (by the OC timing stability of $P_2$).
    
    We sample $(r_1, r_1') \sim (\Tilde{r}_1, \Tilde{r}_1')$ and sample $\env_2$ conditioned on $(r_1, y)$ and $\env_2'$ conditioned on $(r'_1, y)$. We then samlpe $(r_2, r_2') \sim (\Tilde{r}_2, \Tilde{r}_2')|_{\env_2, \env_2'}$ and let
    \begin{align*}
        \Tilde{r} &= r_1 + r_2 + c \\
        &= \TPP{P_1}{x, \env_1}|_{\out(P_1(x, \env_1)) = y} + \TPP{P_2}{(x, y), \env_2}|_{\out(P_2((x, y), \env_1)) = (y, z)} + c \\
        &= \TPP{(P_2 \otimes P_1)}{x, \env_1}|_{\out((P_2 \otimes P_1)(x, \env_1)) = (y, z)}
    \end{align*}
    and similarly for $\Tilde{r}'$, where $c$ is the constant in Definition~\ref{def:composition-compat}. It follows that
    \begin{align*}
        |\Tilde{r} - \Tilde{r}'| &= |r_1 + r_2 + c - r_1' - r_2' - c| \\
        &= |r_1 - r_1' + r_2 - r_2'| \\
        &\le |r_1 - r_1'| + |r_2 - r_2'| \\
        &\le t_1 + t_2
    \end{align*}
\end{proof}

Lemma~\ref{lemma:composition-oc-stable} shows that we can reason about timing stability when composing DP programs. After analyzing the timing stability of the overall composed program, we can add a single timing delay to protect the release from timing attacks. Alternatively, we can compose timing-private programs such that composed program is also timing-private.

\begin{lemma}[Composition of Timing-Private Programs]\label{lemma:composition-timing-privacy}
Let $\mathcal{C} :(\mathcal{M}\times\mathcal{M})\to\mathcal{M}$ be a valid composition function for a privacy measure $M$; meaning that if $M(X_1, X'_1) \le d_1$ and $M(X_2|X_1 = x, X'_2 | X'_1 = x') \le d_2$ then $M((X_1, X_2), (X'_1, X'_2)) \le C(d_1, d_2)$. Let $P_1:\mathcal{X}\times\mathcal{E}\to(\mathcal{X}\times\mathcal{Y})\times\mathcal{E}$ be $(d_0\mapsto d_1)$-timing private with respect to its second output coordinate, input metric $d_\mathcal{X}$, and privacy measure $M$. Let $P_2:(\mathcal{X}\times\mathcal{Y})\times\mathcal{E}\to(\mathcal{Y}\times\mathcal{Z})\times\mathcal{E}$ be $(d_0\mapsto d_2)$-timing private with respect to its second output coordinate, input metric $d_\mathcal{X}$, and privacy measure $M$. If $P_1$ and $P_2$ be composition compatible, then $P_1 \otimes P_2: \mathcal{X}\times\mathcal{E}\to(\mathcal{Y}\times\mathcal{Z})\times\mathcal{E}$ is $(d_0 \mapsto C(d_1, d_2))$-timing-private with respect to $d_\mathcal{X}$ and privacy measure $M$. 
\end{lemma}
\begin{proof}
    By composition-compatibility, for all $x \in \mathcal{X}$ and input-compatible execution environments $\env_1$, there exists an input-compatible execution environment $\env_2$ such that
    \begin{align*}
        \TPP{P_2\otimes P_1}{x, \env_1} = \TPP{P_1}{x, \env_1} + \TPP{P_2}{(x, \out(P_1(x, \env_1))), \env_2} + c
    \end{align*}
    We can therefore analyze  $\TPP{P_2\otimes P_1}{x, \env_1}|_{\out((P_2\otimes P_1)(x, \env_1) = (y, z)}$ using the joint random variable:
    \begin{align*}
        (\TPP{P_1}{x, \env_1}|_{\out(P_1(x, \env_1)) = (x, y)},  \TPP{P_2}{(x, y), \env_2}|_{\out(P_2((x, y), \env_2)) = (y, z)})
    \end{align*}
    
    and similarly for $x'$. Since $P_1$ is timing-private with respect to input metric $d_\mathcal{X}$ and privacy measure $M$, it follows that for all $d_1$-close inputs $x, x' \in\mathcal{X}$, all input-compatible execution environments $\env_1, \env'_1$, and all $y \in \supp(\out(P_1(x, \env_1)))\cap \supp(\out(P_1(x', \env'_1)))$
    \begin{align*}
        M(\TPP{P_1}{x,\env_1}|_{\out(P_1(x,\env_1)) = (x, y)}, \TPP{P_1}{x',\env'_1}|_{\out(P_1(x',\env'_1)) = (x', y)}) \le d_1
    \end{align*}

    Similarly, since $P_2$ is timing-private with respect to input metric $d_\mathcal{X}$ and privacy measure $M$, it follows that for all inputs $(x, y), (x', y) \in\mathcal{X}\times\mathcal{Y}$ that are $d_1$-close on their first coordinate, all input-compatible execution environments $\env_2, \env'_2$, and all $(y, z) \in \supp(\out(P_2((x, y), \env_2)))\cap \supp(\out(P_2((x', y), \env'_2)))$
    \begin{align*}
        M(\TPP{P_2}{(x, y),\env_2}|_{\out(P_2((x,y),\env_2)) = (y,z)}, \TPP{P_2}{(x', y),\env'_2}|_{\out(P_2((x', y),\env'_2)) = (y, z)}) \le d_2
    \end{align*}
    
    Let \begin{align*}
        &Y = \TPP{P_1}{x,\env_1}|_{\out(P_1(x,\env_1)) = (x, y)} \\
        &Y' = \TPP{P_1}{x',\env'_1}|_{\out(P_1(x',\env'_1)) = (x', y)} \\
        &Z = \TPP{P_2}{(x, y),\env_2}|_{\out(P_2((x,y),\env_2)) = (y,z)} \\
        &Z' = \TPP{P_2}{(x', y),\env'_2}|_{\out(P_2((x', y),\env'_2)) = (y, z)}
    \end{align*}
    
    Then 
    \begin{align*}
        M((Y, Z), (Y', Z')) &\le C(M(Y, Y'), M(Z, Z')) \\
        &\le C(d_1, d_2)
    \end{align*}
    which is what we wanted to show.
\end{proof}

The above composition theorem says that we can compose timing-private programs and keep track of the overall timing privacy guarantees. For example, $C((\varepsilon_1, \delta_1), (\varepsilon_2, \delta_2)) = (\varepsilon_1 + \varepsilon_2, \delta_1 + \delta_2)$ is a valid composition function for approximate-DP (Lemma~\ref{lem:dp composition}). Lemma~\ref{lemma:composition-timing-privacy} says that the composition of a program $P_1$ that is $(\varepsilon_1, \delta_1)$-timing private with a program $P_2$ that is $(\varepsilon_2, \delta_2)$-timing private, results in a new program that is $(\varepsilon_1 + \varepsilon_2, \delta_1 + \delta_2)$-timing private. 

We can use the above fact to construct timing-private programs, for example, that compute DP means from a timing-private DP sum program and a timing-private DP count program. The DP sum and DP count can be interpreted as the numerator and denominator of the mean, respectively. For example, Program~\ref{program:dataset-count} accepts as input $(x, y)$ where $x$ is a dataset and $y$ is a DP sum over $x$. By chaining together Program~\ref{program:dataset-count} with a (modified) Program~\ref{program:discrete-laplace}, we obtain a program that outputs a DP sum followed by a DP count. 

\begin{algorithm}[t]
\vspace{5px}
\begin{flushleft}
    
\textbf{Input:} A dataset $x$ occupying memory locations $M[0],\dots, M[\texttt{input\_len} - 2]$ and a value $y\in\mathbb{N}$ at memory location $M[\texttt{input\_len} - 1]$. The value $y$ can be a (possibly DP) sum over the elements $M[0],\dots,M[\texttt{input\_len} - 1]$, for example.  

\vspace{5px}
\textbf{Output:} The value $y$ and the size of the dataset $\texttt{input\_len} - 1$. 
\end{flushleft}
\vspace{5px}
\begin{algorithmic}[1]
\STATE $\texttt{output\_ptr} = \texttt{input\_len} - 2$;
\STATE $\texttt{output\_len} = 2$;
\STATE $M[\texttt{output\_ptr} + 1] = \texttt{input\_len} - 1$;
\STATE $\texttt{HALT}$;
\end{algorithmic}

\caption{Dataset Count Program}\label{program:dataset-count}
\end{algorithm}

\subsection{A Timing Private Unbounded Mean}
We now give an example of how one can use composition-compatible programs to compute timing-private DP means in the unbounded DP model. 

\begin{remark}\label{remark:modified-programs}
    In the following constructions, we will assume a modified version of the timing-private DP sum program from Lemma~\ref{lemma:timing-private-sum} that instead writes its output in a composition-compatible way (by appending its output to the end of its input and setting $\texttt{output\_ptr}$ and $\texttt{output\_len}$ accordingly). This is a straight-forward modification that involves updating Program~\ref{program:sum} and Program~\ref{program:discrete-laplace} so that the original input is included in the chained program's output. Similarly, when we chain together Program~\ref{program:discrete-laplace} with Program~\ref{program:dataset-count}, we will assume a modified version of the Discrete Laplace program that accepts as input $(y_1, y_2)$ and outputs $(y_1, y_2 + \eta)$ where $\eta$ is the noise drawn from the Discrete Laplace distribution. We note that these changes can be made such that the stability claims for Lemma~\ref{lemma:oc-stability-disclap} and  Lemma~\ref{lemma:joint-stability-sum} continue to hold by ensuring that the modifications incur some constant additive overhead to the runtime.
\end{remark}

\begin{lemma}\label{lemma:joint-count-stability}
    The Dataset Count program $P:(\mathcal{X}\times\mathbb{N})\times\mathcal{E}\to(\mathbb{N}\times\mathbb{N})\times\mathcal{E}$ (Program~\ref{program:dataset-count}) is $(\din\to \{\din, 0\})$-jointly output/timing stable under the output distance metric $d_\mathbb{N}$ on the second output coordinate of $(\mathbb{N}\times\mathbb{N})$ and the input distance metric $\did$.
\end{lemma}
\begin{proof}
    The program always executes in $4$ instructions and is therefore $(\din\mapsto 0)$-timing stable under the input distance metric $\did$ (Lemma~\ref{lemma:input-independence}). If you add or remove a row from the dataset $x \in \mathcal{X}$, then the count changes by at most $\din$. Therefore, the program is $(\din\mapsto \din)$-output stable (with respect to its second output coordinate) under the output distance metric $d_\mathbb{N}$ and input metric $\did$. Since the program is deterministic in its output, the program satisfies $(\din \mapsto \{\din, 0\})$-joint output/timing stability (Lemma~\ref{lemma:det-joint-stability}).
\end{proof}

\begin{lemma}\label{lemma:noisy-count-output-private}
   Let $P_1:(\mathcal{X}\times\mathbb{N})\times\mathcal{E}\to(\mathbb{N}\times\mathbb{N})\times\mathcal{E}$ be the Dataset Count program (Program~\ref{program:dataset-count}) and $P_2: (\mathbb{N}\times\mathbb{N}) \times\mathcal{E} \to (\mathbb{N}\times\mathbb{N})\times\mathcal{E}$ be the Discrete Laplace program (Program~\ref{program:discrete-laplace}) modified according to Remark~\ref{remark:modified-programs} with $s = 1/\varepsilon_2 = 1/\ln(\texttt{b}/\texttt{a})$. Then $P_2 \circ P_1: (\mathcal{X}\times\mathbb{N})\times\mathcal{E}\to(\mathbb{N}\times\mathbb{N})\times\mathcal{E}$ is $\varepsilon_2$-differentially private under the max-divergence privacy measure with respect to the 2nd coordinate of $(\mathbb{N}\times\mathbb{N})$ and the input distance $\did$.
\end{lemma}
\begin{proof}
   $P_1$ is $(\din\mapsto \din)$-output stable with respect to input distance metric $\did$ and output distance metric $d_\mathbb{N}$ on the second output coordinate of $(\mathbb{N}\times\mathbb{N})$ (Lemma~\ref{lemma:joint-count-stability} and Lemma~\ref{lemma:joint-implies-output-stable}). Therefore the proof follows exactly by Lemma~\ref{lemma:laplace-output-privacy} since $P_2 \circ P_1$ is the Discrete Laplace program executed on $\din$-close inputs under the input distance $d_\mathbb{N}(n, n') = |n - n'|$.
\end{proof}

\begin{lemma}\label{lemma:oc-stability-noisy-count}
    Let $P_1:(\mathcal{X}\times\mathbb{N})\times\mathcal{E}\to(\mathbb{N}\times\mathbb{N})\times\mathcal{E}$ be the Dataset Count program (Program~\ref{program:dataset-count}), and $P_2: (\mathbb{N}\times\mathbb{N})\to (\mathbb{N}\times\mathbb{N})\times\mathcal{E}$ be the Discrete Laplace program (Program~\ref{program:discrete-laplace}) modified according to Remark~\ref{remark:modified-programs}. Then the chained program $P_2 \circ P_1: (\mathcal{X}\times\mathbb{N})\to(\mathbb{N}\times\mathbb{N})\times\mathcal{E}$ is $(\din\mapsto 5\cdot\din)$-OC timing stable under the insert-delete distance metric $\did$. 
\end{lemma}

\begin{proof}
    The dataset count program is constant-time and therefore $(\din\mapsto\{\din, 0\})$-jointly output/timing stable under $\did$. By Lemma~\ref{lemma:oc-stability-disclap} and Lemma~\ref{lemma:oc-chaining}, we have that the chained program is $(\din\mapsto 5\cdot\din)$-OC timing stable under $\did$.
\end{proof}

\begin{lemma}\label{lemma:timing-private-dp-count}
    Let $P_1:(\mathcal{X}\times\mathbb{N})\times\mathcal{E}\to(\mathbb{N}\times\mathbb{N})\times\mathcal{E}$ be the Dataset Count program (Program~\ref{program:dataset-count}), and $P_2: (\mathbb{N}\times\mathbb{N})\to (\mathbb{N}\times\mathbb{N})\times\mathcal{E}$ be the Discrete Laplace program (Program~\ref{program:discrete-laplace}) modified according to Remark~\ref{remark:modified-programs}, with scale parameter $s_1 = 1/\varepsilon_1$. Let  $P_3:\mathcal{Z}\times\mathcal{E}\to\mathcal{Z}\times\mathcal{E}$ be Program~\ref{program:disclap-timing-private} (Timing-Private Delay) with scale parameter $s_2 = 5/\varepsilon_2$ and shift $\mu$. Then $P_3 \circ (P_2 \circ P_1)$ is $\varepsilon_1$-differentially private in its output and $(\varepsilon_2, \delta)$-timing private for $\delta = 2\cdot e^{\varepsilon_2(\mu - 5)/5}$.
\end{lemma}

\begin{proof}
    The chained program $P_2\circ P_1 : \mathcal{X}\times\mathcal{E}\to\mathbb{N}\times\mathcal{E}$ is $\varepsilon_1$-differentially private by Lemma~\ref{lemma:noisy-count-output-private}. Since $P_3$ is a timing-private delay program (Lemma~\ref{lemma:disclap-timing-private-delay}) it will compute a post-processing (identity function) on the output of $P_2 \circ P_1$ and by Lemma~\ref{lemma:timing-private-delay-post-process} the  chained program $P_3 \circ (P_2 \circ P_1):\mathcal{X}\times\mathcal{E}\to\mathbb{N}\times\mathcal{E}$ is $\varepsilon_1$-differentially private.
    
    Since $P_2 \circ P_1$ is $(1\mapsto 5)$-OC timing stable (Lemma~\ref{lemma:oc-stability-noisy-count}), and $P_3$ is a $(5\mapsto (\varepsilon_2, \delta))$-timing-private delay program for $\delta = 2\cdot e^{-\varepsilon_2(\mu - 5)/5}$ (Lemma~\ref{lemma:disclap-timing-private-delay}) under input distance metric $d_\mathbb{N}$ and the smoothed max-divergence privacy measure $D_{\infty}^\delta$, then by Theorem~\ref{thm:main-theorem}, the chained program $P_3 \circ(P_2 \circ P_1)$ is $(1\mapsto (\varepsilon_2, \delta))$-timing-private with respect to privacy measure $D_\infty^\delta$.
\end{proof}

\begin{lemma}\label{lemma:noisy-mean-compcompat}
    Let $P_1:\mathcal{X}\times\mathcal{E}\to(\mathcal{X}\times\mathcal{Y})\times\mathcal{E}$ be the $\varepsilon_1$-DP and $(\varepsilon_2,\delta)$-timing-private sum program (Lemma~\ref{lemma:timing-private-sum}) that has been modified according to Remark~\ref{remark:modified-programs}. Let $P_2: (\mathcal{X}\times\mathbb{N})\times\mathcal{E}\to(\mathbb{N}\times\mathbb{N})\times\mathcal{E}$ be the $\varepsilon_1$-DP and $(\varepsilon_2, \delta)$-timing-private dataset count program (Lemma~\ref{lemma:timing-private-dp-count}). Then $P_1$ and $P_2$ are composition-compatible.
\end{lemma}

\begin{proof}
    Observe that, due to the modifications from Remark~\ref{remark:modified-programs}, $P_1$ writes to its output the original dataset $x\in\mathcal{X}$ followed by the DP sum $y\in\mathbb{N}$. By Lemma~\ref{lemma:comp-compat-word-ram} the claim follows. 
\end{proof}

\begin{lemma}\label{lemma:noisy-mean-composition}
    Let $P_1:\mathcal{X}\times\mathcal{E}\to(\mathcal{X}\times\mathcal{Y})\times\mathcal{E}$ be the $\varepsilon_1$-DP and $(\varepsilon_2,\delta)$-timing-private sum program (Lemma~\ref{lemma:timing-private-sum}) that has been modified according to Remark~\ref{remark:modified-programs}. Let $P_2: (\mathcal{X}\times\mathbb{N})\times\mathcal{E}\to(\mathbb{N}\times\mathbb{N})\times\mathcal{E}$ be the $\varepsilon_1$-DP and $(\varepsilon_2, \delta)$-timing-private dataset count program (Lemma~\ref{lemma:timing-private-dp-count}). Then $P_2\otimes P_1:\mathcal{X}\times\mathcal{E}\to(\mathbb{N}\times\mathbb{N})\times\mathcal{E}$ is $2\varepsilon_1$-DP and $(2\varepsilon_2, 2\delta)$-timing-private.
\end{lemma}

\begin{proof}
    The first output coordinate of $P_2\otimes P_1$ is the output of $P_1$, and the second output coordinate of $P_2\otimes P_1$ is the output of $P_2$. Since $P_1$ is $\varepsilon_1$-DP under input metric $\did$ and output metric $d_\mathbb{N}$ (Lemma~\ref{lemma:noisy-sum-output-private}), and $P_2$ is $\varepsilon_1$-DP under input metric $\did$ and output metric $d_\mathbb{N}$ (Lemma~\ref{lemma:noisy-count-output-private}), it follows that for all $x, x' \in \mathcal{X}$ satisfying $\did(x, x')\le 1$, all input-compatible $\env_1, \env'_1, \env_2, \env_2' \in \mathcal{E}$, and all $(S_1, S_2) \subseteq (\mathbb{N}\times\mathbb{N})$:
    \begin{align*}
        \Pr[\out((P_2\otimes P_1)(x, \env_1)) \in (S_1, S_2)] & \\
        &\hspace{-150pt}= \Pr[\out(P_1(x, \env_1)) \in S_1]\cdot \Pr[\out(P_2((x, \out(P_1(x, \env_1))), \env_2)) \in S_2] \\
        &\hspace{-150pt}\le e^{\varepsilon_1} \cdot \Pr[\out(P_1(x', \env_1')) \in S_1]\cdot e^{\varepsilon_2}\cdot \Pr[\out(P_2((x', \out(P_1(x', \env_1'))), \env_2')) \in S_2] \\
        &\hspace{-150pt}\le e^{\varepsilon_1 + \varepsilon_2} \cdot \Pr[\out((P_2 \otimes P_1)(x', \env_1')) \in (S_1, S_2)]
    \end{align*}

    Therefore, $P_2 \otimes P_1$ is $2\varepsilon_1$-DP with respect to input metric $\did$ and the max-divergence privacy measure. The claim that $P_2 \otimes P_1$ is $(2\varepsilon_2, 2\delta)$-timing-private follows from the fact that $P_1$ is $(\varepsilon_2, \delta)$-timing-private, $P_2$ is $(\varepsilon_2, \delta)$-timing-private, and Lemma~\ref{lemma:composition-timing-privacy} and letting $C((\varepsilon, \delta), (\varepsilon, \delta)) = (2\varepsilon, 2\delta)$ be the composition function by Lemma~\ref{lem:dp composition}.  
\end{proof}

\section{Implementation}
A proof-of-concept implementation of our framework on top of the OpenDP library in available in the Github repository~\cite{Shoemate_OpenDP_Library}.  
 The proof of concept supports defining timing stability maps for transformations and output-conditional timing stability maps for measurements. We additionally implemented a proof-of-concept timing-private delay function that can be chained with arbitrary output-conditional timing stable measurements to delay their output release. 

The purpose of the implementation is to illustrate the compatibility of our framework with existing differential privacy libraries, not to claim that the implementation provides timing privacy for physical executions.  
As discussed in Section~\ref{sec:future-work}, we leave for future work the problem of instantiating our framework for physical timing channels, which would involve constraining the execution environment, identifying the appropriate units to measure timing, and finding realistic upper bounds on actual timing-stability constants. 
\section{Acknowledgments}
The authors would like to thank Mike Shoemate for his help with implementing the proof-of-concept in the OpenDP library. 

Zachary Ratliff's work was supported in part by Cooperative Agreement CB20ADR0160001 with the Census Bureau, Salil Vadhan's Simons Investigator Award, and a visit to the School of Computer Science at the University of Sydney.

Salil Vadhan is supported by a Simons Investigator Award, Cooperative Agreement CB20ADR0160001 with the Census Bureau, a grant from the Sloan Foundation, and gifts from Apple and Mozilla. This work was completed in part while he was a Visiting Researcher to the School of Computer Science at the University of Sydney and the Sydney Mathematical Research Institute.
\bibliography{refs}
\appendix 
\newpage
\section{Stability Proofs}\label{app:proofs}
\begin{lemma}\label{lemma:input-independence}
A program $P:\mathcal{X}\times\mathcal{E}\to\mathcal{Y}\times\mathcal{E}$ is $(\din\mapsto 0)$-\emph{timing-stable} if for all pairs of inputs $x, x' \in \mathcal{X}$, and all input-compatible execution environments $\env, \env'\in\mathcal{E}$, $\TP{x, \env} \equiv \TP{x', \env'}$.
\end{lemma}

\begin{proof}
For any two inputs $x$ and $x'$, take the coupling $(\Tilde{r}, \Tilde{r}')$ to be $(\TP{x, \env}, \TP{x, \env})$. Since $\TP{x, \env} \equiv \TP{x', \env'}$, it follows that the marginal distributions of $\Tilde{r}$ and $\Tilde{r}'$ are identical to $\TP{x, \env}$ and $\TP{x', \env'}$, and $|\Tilde{r} - \Tilde{r}'| = 0$. 
\end{proof}

\begin{lemma}\label{lemma:constant-time-stable}
If $P: \mathcal{X}\times\mathcal{E} \to \mathcal{Y}\times\mathcal{E}$ satisfies for all $x \in \mathcal{X}$, $\env\in\mathcal{E}$, $\TP{x, \env} = c$ for some constant $c$, then $P$ is $(\din \mapsto 0)$-\emph{timing stable}.
\end{lemma}
\begin{proof}
Since $P$ is constant time, $\TP{x, \env} = c$ for all $x$ and $\env$, then for all $x, x'$ and input-compatible $\env, \env'$, $\TP{x, \env} \equiv \TP{x', \env'}$ and the claim follows from Lemma~\ref{lemma:input-independence}.
\end{proof}

\begin{lemma}\label{lemma:rt-implies-oc}
If $P: \mathcal{X}\times\mathcal{E} \to \mathcal{Y}\times\mathcal{E}$ is deterministic (in its output) and is $(\din \mapsto \tout)$-\emph{timing stable}, then $P$ is $(\din \mapsto \tout)$-\emph{OC-timing stable}.
\end{lemma}

\begin{proof}
    For all $x$ and $x'$ satisfying $d_{\mathcal{X}}(x, x') \le \din$, and any pair of environments $\env, \env' \in\mathcal{E}$, either $\out(P(x, \env)) \ne \out(P(x', \env'))$ or $\out(P(x, \env)) = \out(P(x', \env'))$. If $\out(P(x, \env)) \ne \out(P(x', \env'))$, then there is no requirement on the distributional closeness of $\TP{x, \env}$ and $\TP{x', \env'}$ and the claim is satisfied. Now suppose that $\out(P(x, \env)) = \out(P(x',\env')) = y$. By the timing stability of $P$ there exists a coupling $(\Tilde{r}, \Tilde{r}')$ of the random variables $\TP{x, \env}$ and $\TP{x', \env'}$ such that $|\Tilde{r} - \Tilde{r}'| \le \tout$. Since $P$ is deterministic it also follows that $\Tilde{r}$ and $\Tilde{r}'$ have the same marginal distributions as $\TP{x, \env}_{|\out(P(x, \env)) = y}$ and $\TP{x'}_{|\out(P(x', \env')) = y}$ respectively, and the claim holds.
\end{proof}

\begin{lemma}\label{lemma:constant-implies-oc}
If $P: \mathcal{X}\times\mathcal{E} \to \mathcal{Y}\times\mathcal{E}$ has constant runtime then $P$ is $(\din \to 0)$-\emph{OC timing stable}.
\end{lemma}
\begin{proof}
Since $P$ is constant time, $\TP{x, \env} = c$ for all $x$ and $\env$. For all $y \in \supp(\out(P(x, \env))) \cap \supp(\out(P(x', \env')))$ with $d_{\mathcal{X}}(x, x') \le \din$, let the coupling $(\Tilde{r}, \Tilde{r}')$ of $\TP{x, \env}_{|\out(P(x, \env)) = y}$ and $\TP{x', \env'}_{|\out(P(x', \env')) = y}$ take the value $(c, c)$. Then $\Tilde{r}$ and $\Tilde{r}'$ are identically distributed to $\TP{x, \env}_{|\out(P(x, \env)) = y}$ and $\TP{x', \env'}_{|\out(P(x', \env')) = y}$ respectively, and $|\Tilde{r} - \Tilde{r}'| = |c - c| = 0$.
\end{proof}

\begin{lemma}\label{lemma:joint-implies-timing-stable}
    Jointly output/timing-stable programs are timing-stable programs.
\end{lemma}

\begin{proof}
    If $P:\mathcal{X}\times\mathcal{E}\to\mathcal{Y}\times\mathcal{E}$ is $(\din\mapsto \{\dout, \tout\})$-jointly output/timing stable with respect to input metric $d_\mathcal{X}$, then for every $x, x' \in\mathcal{X}$ satisfying $d_\mathcal{X}(x, x') \le \din$ and all pairs of execution environments $\env, \env' \in \mathcal{E}$, there exists a coupling $((\Tilde{u}, \Tilde{r}), (\Tilde{u}', \Tilde{r}'))$ of the joint random variables $(\out(P(x, \env)), \TP{x, \env})$ and $(\out(P(x', \env')), \TP{x', \env'})$ satisfying $\Pr[|\Tilde{r} - \Tilde{r}'| < \tout] = 1$. Take $(\Tilde{r}, \Tilde{r}')$ to be the coupling of the random variables $\TP{x, \env}$ and $\TP{x', \env'}$ and the claim follows.
\end{proof}

\begin{lemma}\label{lemma:joint-implies-output-stable}
    Jointly output/timing-stable programs are output-stable programs.
\end{lemma}

\begin{proof}
    If $P:\mathcal{X}\times\mathcal{E}\to\mathcal{Y}\times\mathcal{E}$ is $(\din\mapsto \{\dout, \tout\})$-jointly output/timing stable with respect to input metric $d_\mathcal{X}$ and output metric $d_\mathcal{Y}$, then for every $x, x' \in\mathcal{X}$ satisfying $d_\mathcal{X}(x, x') \le \din$ and all pairs of execution environments $\env, \env' \in \mathcal{E}$, there exists a coupling $((\Tilde{u}, \Tilde{r}), (\Tilde{u}', \Tilde{r}'))$ of the joint random variables $(\out(P(x, \env)), \TP{x, \env})$ and $(\out(P(x', \env')), \TP{x', \env'})$ satisfying $d_\mathcal{Y}(\Tilde{u}, \Tilde{u}') \le \dout$ with probability $1$. Take $(\Tilde{u}, \Tilde{u}')$ to be the coupling of the random variables $\out(P(x, \env))$ and $\out(P(x', \env'))$ and the claim follows.
\end{proof}
\end{document}